\documentclass[12 pt]{amsart}%amsart

\usepackage{amsthm}
\usepackage{amssymb}
\usepackage{latexsym}
\usepackage{url}
\usepackage{amsmath}
\usepackage{verbatim}
\usepackage{graphicx}
\usepackage{epsf}
\usepackage{enumerate}
\usepackage{geometry}
\usepackage[width=6cm]{caption}
\usepackage{hyperref}
\hypersetup{breaklinks=true,
pagecolor=white,
colorlinks=true,%false
citecolor=black,%
filecolor=black,%
linkcolor=black,%
urlcolor=black
}

\newtheorem{thm}{Theorem}[section]

\newtheorem{lem}[thm]{Lemma}
\newtheorem{prop}[thm]{Proposition}

\newtheorem{defin}[thm]{Definition}

\theoremstyle{definition}
\newtheorem{expl}[thm]{Example}
\newtheorem{remark}[thm]{Remark}

\newtheorem{question}[thm]{Question}
\newtheorem{answer}[thm]{Answer}

\newcommand{\dom}{\textnormal{dom}}

\newcommand{\R}{\mathbb{R}}

\newcommand{\N}{\mathbb{N}}

\newcommand{\wt}{\widetilde}

\title[The geometric stability of Voronoi diagrams]{The geometric stability of Voronoi diagrams with respect to small changes of the sites}
\author{Daniel Reem}
\thanks{
\noindent Department of Mathematics, University of Haifa, Mount Carmel, 31905 Haifa, Israel. A large part of this research was done while the author was in the  Department of Mathematics at the Technion - Israel Institute of Technology,  Haifa, Israel.\\
\noindent E-mail: dream@tx.technion.ac.il , dream@math.haifa.ac.il %,$\,\,$ %\url{http://www.technion.ac.il/~dream}
}

\begin{document}
\maketitle
\begin{abstract}
Voronoi diagrams appear in many areas in science and technology and have numerous  applications. They have been the subject of extensive investigation during the last decades. Roughly speaking, they are a certain decomposition of a given space into cells, induced by a distance function and by a tuple of subsets called the generators or the sites.  
 Consider the following question: does a small change of the  sites, e.g., of their position or shape, yield a small change in  the corresponding Voronoi cells?  This  question is by all means natural and fundamental, since in practice one approximates the sites either because of  inexact  information about them, because of inevitable numerical errors in their representation,  for simplification  purposes and so on, and it is important to know whether the  resulting Voronoi cells approximate the real ones well. The traditional approach to Voronoi diagrams, and, in particular, to (variants of) this question, is combinatorial. However, it seems that there has been a very limited discussion in the geometric sense (the shape of the cells),  mainly an intuitive one, without proofs, in Euclidean spaces. 
We formalize this question precisely, and then show that the answer is positive in the case of $\R^d$, or, more generally, in (possibly infinite dimensional) uniformly convex  normed spaces, assuming there is a common  positive lower bound on the distance between the sites. Explicit bounds are given, and  we allow infinitely many sites of a general form.  The relevance of this result is illustrated using several pictures and many real-world and theoretical examples and counterexamples. 
\end{abstract}
\newpage

\section{Introduction}\label{sec:Intro}
\subsection{Background} The Voronoi diagram (the Voronoi tessellation, the  Voronoi decomposition, the Dirichlet tessellation) is  one of the basic structures
in computational geometry. Roughly speaking, it is a certain decomposition of a given space $X$ into cells, induced by a distance function and by a tuple of subsets $(P_k)_{k\in K}$, called the generators or the sites. More precisely, the Voronoi cell $R_k$ associated with the site $P_k$ is the set
of all the points in $X$ whose distance to $P_k$ is not greater than their distance to the union of the other sites $P_j$.

 Voronoi diagrams appear in a huge number of fields in  science and technology and have  many applications. They have been the subject of   research for at least 160 years, starting formally with L. Dirichlet \cite{Dirichlet} and G. Voronoi \cite{Voronoi}, and of extensive research during the last 40  years.  For several well written surveys on Voronoi diagrams which contain extensive bibliographies and many applications, see  \cite{Aurenhammer}, \cite{AurenhammerKlein}, \cite{OBSC}, and  \cite{VoronoiWeb}.\\

\noindent Consider the following question: 
\begin{question}\label{ques:main}
Does a small change of the  sites, e.g., of their position or shape, yield a small change in the corresponding Voronoi cells?  
\end{question}
This  question is by all means natural, because in practice, no matter which  algorithm is being used for the computation of the Voronoi cells, one  approximates  the sites either because of lack of exact information about them, because of inevitable numerical errors occurring when a site is represented in an analog or a digital manner, for simplification  purposes and so on, and it is important to know whether the resulting Voronoi cells approximate well the real ones.   

For instance, consider the Voronoi diagram whose sites are either 
shops (or large shopping centers),  antennas, or other facilities in some city/district such as post offices. See Figures \ref{fig:Shops}-\ref{fig:ShopsReality}. 
%%%%%%%%%%%%%%%%%%%%%%%%%%%%%%%%%%%%%%%%%%%%%%%%%%%%%%%%%%%%%%%%%%%%%%%
\begin{figure*}
\begin{minipage}[t]{0.5\textwidth}
\begin{center}
{\includegraphics[scale=0.78]{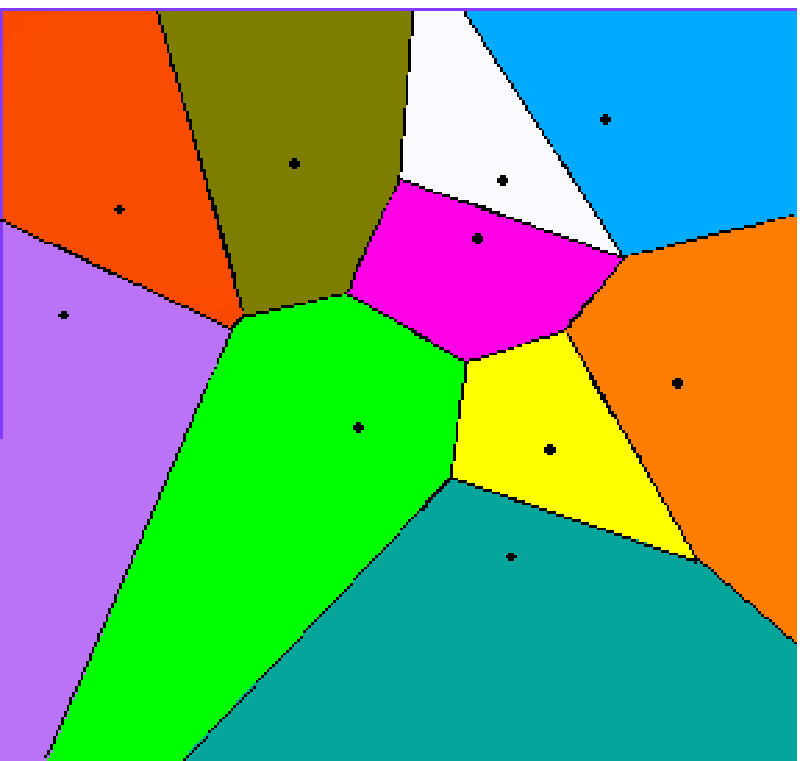}}
\end{center}
 \caption{10 shopping centers (or post offices) in a flat city. Each shopping center is represented by a point. }
\label{fig:Shops} %\label{page:Stabilityl3}
\end{minipage}%
\hfill
\begin{minipage}[t]{0.5\textwidth}
\begin{center}
{\includegraphics[scale=0.78]{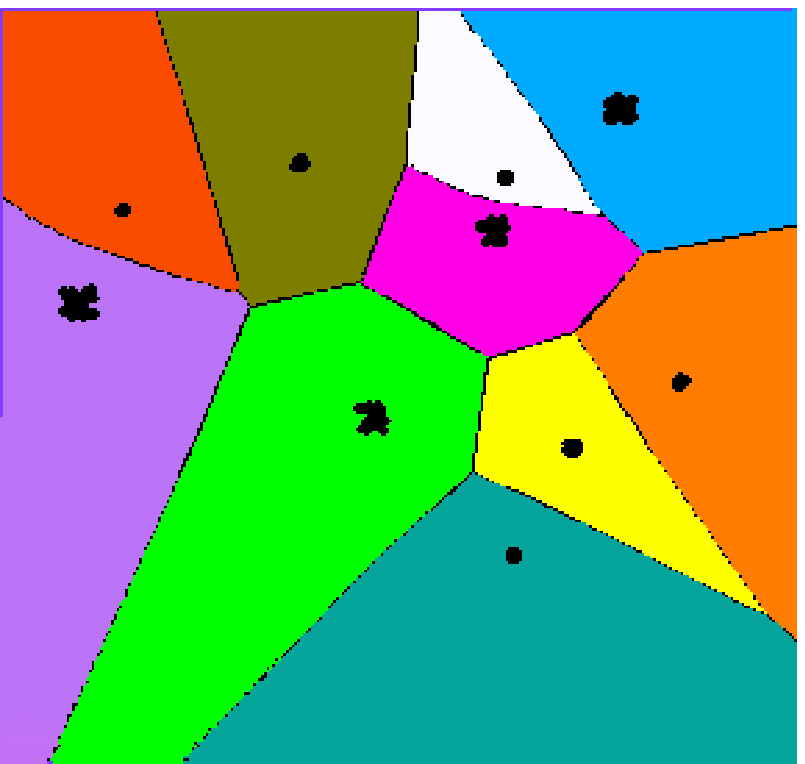}}
\end{center}
 \caption{In reality each shopping center/post office  is not a point and its location is approximated. The combinatorial structure is somewhat different and the Voronoi cells are not exactly polygons, but still, their shapes are almost the same as in Figure \ref{fig:Shops}.}
\label{fig:ShopsReality}
\end{minipage}%
\end{figure*}
%%%%%%%%%%%%%%%%%%%%%%%%%%%%%%%%%%%%%%%%%%%%%%%%%%%%%%%%%%%%%%%

Each Voronoi cell is the domain of influence  of its site and it can be used for various purposes, among them  estimating the number of potential costumers \cite{EconomyFacilityPNAS} or  understanding the spreading patterns of mobile phone viruses \cite{VoronoiVirus}. In reality, each site has a somewhat vague shape, and its real location is not known exactly. However, to simplify matters we regard each  site as a point (or a finite collection of points if we consider firms of shops) located more or less near the real location. As a result, the resulting cells only approximate the real ones, but we hope that the approximation will be good in the geometric sense, i.e., that the shapes of the corresponding real and approximate cells will be almost the same. (See Section \ref{sec:examples} for many additional examples, including ones with infinitely many sites or in higher/infinite dimensional spaces.) As the counterexamples in  Section  \ref{sec:CounterExamples} show, it is  definitely not obvious that this is the case. 

A similar question to Question \ref{ques:main} can be asked regarding any geometric structure/algorithm, and, in our opinion, it is a fundamental question which is analogous to the question about the stability of the solution of a differential equation with respect to small changes in the initial conditions. 

The traditional approach to Voronoi diagrams, and, in particular, to  (variants of) Question \ref{ques:main}, is combinatorial.  For instance, as already mentioned in Aurenhammer \cite[p. 366]{Aurenhammer}, the combinatorial structure of Voronoi diagrams (in the case of the Euclidean distance with point sites), i.e., the structure of vertices, edges and so 
on, is not stable under continuous motion of the sites, but it is stable ``most of the time''. A more extensive discussion about this issue, still with point sites but possibly in higher dimensions, can be found in Weller \cite{Weller}, Vyalyi et al. \cite{VGT}, and Albers et al.  \cite{AGMR}. 

However, it seems that this question, in the geometric sense, has been raised or discussed only rarely in the context of Voronoi diagrams. In fact, after a comprehensive survey of the literature about Voronoi diagrams, we have found only very few places that have a very brief, particular, and intuitive discussion which  is somewhat related to  this question. The following cases were mentioned: the Euclidean plane with finitely many sites \cite{Kaplan}, the Euclidean plane with finitely many point sites  
\cite[p. 366]{Aurenhammer}, and the $d$-dimensional Euclidean space with finitely many point sites \cite{AGMR} (see also some of the references therein). It was claimed there without proofs and exact definitions that the  Voronoi cells have a continuity property: a small change in the position or the shape of the sites yields a small change in the corresponding Voronoi cells.

Another continuity property was discussed by Groemer \cite{Groemer} in the  context of the geometry of numbers. He considered Voronoi diagrams generated by a lattice of points in a $d$-dimensional Euclidean space, and proved that if a sequence of  lattices converges to a certain lattice (meaning that the basis elements which generate the lattices converge with respect to the Euclidean distance to the basis which  generates the limit lattice), then the corresponding Voronoi cells of the origin converge, with respect to the Hausdorff distance, to the cell of the origin of the limit lattice. His result is, in a sense and in a very particular case, a stability result, but it definitely does not answer Question \ref{ques:main} (which, actually, was  not asked at all in \cite{Groemer}) for several reasons:  first, usually the sites or the perturbed ones do not form a (infinite) lattice. Second, in many cases they are not points (singletons). Third, a site is usually different from the perturbed site (in \cite{Groemer} the discussed sites equal $\{0\}$). In this connection, we also note that Groemer's proof is very restricted to the above setting and it uses arguments based on compactness and no explicit bounds are given.

It is quite common in the computational geometry literature to assume ``ideal conditions'', say infinite precision in the computation, exact and simple  input, and so on. These conditions are somewhat non-realistic. Issues related  to the stability of geometric structures under small perturbations of their building blocks (not necessarily the geometric  stability) are not so common  in  the literature, but they can be found in several places, e.g.,  in  \cite{AGGKKRS, AryaMalamatosMount, AttaliBoissonnatEdelsbrunner, BandyopadhyaySnoeyink, ChazalCohenSteinerLieutier, ChoiSeidel, SteinerEdelsbrunerHarer,  FortuneStability,  HarPeled, Khanban, LofflerPhD, LofflerKreveld, GuibasSalesinStolfi,  SugiharaIriInagakiImai}. However, in many of the above places the discussion has combinatorial characteristics and there are several restrictive assumptions: for instance, the underlying setting is usually a finite dimensional space (in many cases only $\R^2$ or $\R^3$), with the Euclidean distance, and with  finitely many  objects of a specific form (merely points in many cases). In addition, the methods are restricted to this setting. In contrast, the infinite dimensional case or the case of (possibly infinitely many) general objects or general norms have never been considered.

\subsection{Contribution of this paper} We discuss the question of stability of Voronoi diagrams with respect to small changes of the corresponding sites. We first formalize this question precisely, and then show that the answer is positive in the case of $\R^d$, or, more generally, in the case of (possibly infinite dimensional) uniformly convex normed spaces, assuming there is a  common positive lower bound on the distance between the sites. Explicit bounds are presented, and we allow infinitely many sites of a general form. We also present several counterexamples which show that the assumptions  formulated in the main result are crucial. We illustrate the relevance of this result using several pictures and many real-world and theoretical examples and counterexamples. To the best of our knowledge, the main result and the approach used for deriving it are new. Two of our main tools are: a new representation theorem which characterizes the Voronoi cells as a collection of line segments and a new geometric lemma which provides an explicit geometric estimate.   

\subsection{The structure of the paper} In Section \ref{sec:Definitions} we present the basic definitions and notations. Exact formulation of Question  \ref{ques:main} and informal description of the main result are given in Section \ref{sec:FormalInformal}. The relevance of the main result  is illustrated using many theoretical and real-world  examples in Section \ref{sec:examples}. The main result is presented in Section \ref{sec:Outline}, and we discuss briefly some aspects related to its proof. In Section \ref{sec:CounterExamples} we present several interesting counterexamples showing that the assumptions imposed in the main result are crucial. We end the paper in Section \ref{sec:Concluding} with several concluding remarks.  Since the proof of the main result is quite long and technical, and because the main goal of this paper is to introduce the issue and to discuss it in a qualitative manner, rather than going deep into technical details, proofs were omitted from the main body of the text. Full proofs  can be found in the appendix (Section \ref{sec:appendix}) and a preliminary version in \cite{ReemPhD}.

\section{Notation and basic definitions}\label{sec:Definitions}
In this section we present our notation and basic definitions. In the main discussion we consider a closed and convex set $X\neq \emptyset$ in some uniformly convex normed space $(\widetilde{X},|\cdot|)$ (see Definition \ref{def:UniformlyConvex} below), real or complex, finite or infinite dimensional. The induced metric is $d(x,y)=|x-y|$. We assume that $X$ is not a singleton, for otherwise everything is trivial.   We denote by $[p,x]$ and $[p,x)$ the closed and half open line segments connecting $p$ and $x$, i.e., the sets $\{p+t(x-p): t\in [0,1]\}$ and $\{p+t(x-p): t\in [0,1)\}$ respectively. The (possibly empty) boundary of $X$ with respect to the affine hull spanned by $X$ is denoted by $\partial X$. The open ball with center $x\in X$ and radius $r>0$ is denoted by $B(x,r)$. 

\begin{defin}\label{def:dom}
Given two nonempty subsets $P,A\subseteq X$, the dominance region
$\dom(P,A)$ of $P$ with respect to $A$ is the set of all $x\in X$
whose distance to $P$ is not greater than their distance to $A$, i.e.,
\begin{equation*}
\dom(P,A)=\{x\in X: d(x,P)\leq d(x,A)\}.
\end{equation*}
Here $d(x,A)=\inf\{d(x,a): a\in A\}$ and in general we denote  $d(A_1,A_2)=\inf\{d(a_1,a_2): a_1\in A_1,\,a_2\in A_2\}$ for any nonempty  subsets $A_1,A_2$. 
\end{defin}
\begin{defin}\label{def:Voronoi}
Let $K$ be a set of at least 2 elements (indices), possibly
infinite. Given a  tuple $(P_k)_{k\in K}$ of nonempty subsets
$P_k\subseteq X$, called the generators or the sites, the Voronoi diagram  induced by this tuple is the tuple $(R_k)_{k\in K}$ of non-empty subsets
$R_k\subseteq X$, such that for all $k\in K$,
\begin{equation*}
R_k=\dom(P_k,{\underset{j\neq k}\bigcup P_j})
=\{x\in X: d(x,P_k)\leq d(x,P_j)\,\,\forall j\in K ,\, j\neq k \}.
\end{equation*}
 In other words,  the Voronoi cell $R_k$ associated with the site $P_k$ is the set of all $x\in X$ whose distance to $P_k$ is not greater than their distance to the union of the other sites $P_j$.
\end{defin}
In general, the Voronoi diagram induces  a decomposition of $X$   into its Voronoi cells and the rest. If $K$ is finite, then the union of the cells is the whole space. However, if $K$ is infinite, then there may be a ``neutral cell'': for example, if $X$ is the Euclidean plane, $K=\N=\{1,2,3,\ldots\}$ and $P_k=\R\times \{1/k\}$, then no point in the lower half-plane $\R\times (-\infty,0]$ belongs to any  Voronoi cell. In the above definition and the rest of the paper we ignore the neutral cell.

We now recall the definition of strictly and uniformly convex spaces.
\begin{defin}\label{def:UniformlyConvex}\label{page:UniConvDef}
A normed space $(\widetilde{X},|\cdot|)$ is said to be strictly convex if   for all $x,y\in \wt{X}$ satisfying $|x|=|y|=1$ and $x\neq y$, the inequality $|(x+y)/2|<1$ holds. $(\widetilde{X},|\cdot|)$ is said to be uniformly convex if for any $\epsilon\in (0,2]$ there exists $\delta\in (0,1]$ such that for all $x,y\in \wt{X}$, if $|x|=|y|=1$ and $|x-y|\geq \epsilon$, then $|(x+y)/2|\leq 1-\delta$.
\end{defin}
Roughly speaking, if the space is uniformly convex, then for any $\epsilon>0$ there exists a uniform positive lower bound on how deep the midpoint between any two unit vectors must penetrate the unit ball, assuming the distance  between them is at least $\epsilon$. In general normed spaces the penetration is not necessarily positive, since the unit sphere may contain  line segments. $\R^2$ with the max norm $|\cdot|_{\infty}$ is a typical example for this.  A uniformly convex space is always strictly convex, and if it is also  finite dimensional, then the converse is true too. The $m$-dimensional Euclidean  space $\R^m$, or more generally, inner product spaces, the sequence spaces   $\ell_p$, the Lebesgue spaces $L_p(\Omega)$ ($1<p<\infty$),  and a uniformly convex product of a finite number of uniformly convex spaces, are all examples of uniformly convex spaces. See Clarkson \cite{Clarkson} and, for instance, Goebel-Reich \cite{GoebelReich} and Lindenstrauss-Tzafriri \cite{LindenTzafriri} for more information about uniformly convex spaces. 

From the definition of uniformly convex spaces we can obtain a function which
assigns to the given $\epsilon$ a corresponding value $\delta(\epsilon)$. There are several  ways to obtain such a function, but for our purpose we only need $\delta$ to be increasing, and to satisfy $\delta(0)=0$ and $\delta(\epsilon)>0$ for any $\epsilon\in (0,2]$. One choice, which is not necessarily the most convenient one, is the modulus of convexity, which is the function $\delta:[0,2]\to[0,1]$ defined by
\begin{equation*}
\displaystyle{\delta(\epsilon)=\inf\{1-|(x+y)/2|: |x-y|\geq \epsilon,\,|x|=|y|=1\}}.\label{eq:delta}
\end{equation*}
For specific spaces we can take more convenient functions. For instance, for the spaces $L_p(\Omega)$ or $\ell_p\,$, $1<p<\infty$,  we can take
\begin{equation*}
\begin{array}{l}
\delta(\epsilon)=1-(1-\left(\epsilon/2)^p\right)^{1/p},\,\, \textnormal{for}\,\,p\geq 2,\\
\delta(\epsilon)=1-\left(1-(\epsilon/2)^q\right)^{1/q},\,\, \textnormal{for}\,\,1<p\leq 2\,\,
\textnormal{and}\,\,\frac{1}{p}+\frac{1}{q}=1.
\end{array}
\end{equation*}
We finish this section with the definition of the Hausdorff distance, a definition which  is essential for the rest of the paper. 
\begin{defin}\label{def:Hausdorff}
Let $(X,d)$ be a metric space. Given two nonempty sets $A_1,A_2\subseteq X$, the Hausdorff distance between them is defined by
\begin{equation*}
D(A_1,A_2)=\max\{\sup_{a_1\in A_1}d(a_1,A_2),\sup_{a_2\in A_2}d(a_2,A_1)\}.
\end{equation*}
\end{defin}
Note that the Hausdorff distance $D(A_1,A_2)$ is definitely different from the usual distance $d(A_1,A_2)=\inf\{d(a_1,a_2): a_1\in A_1,\,a_2\in A_2\}$. As a matter of fact, $D(A_1,A_2)\leq \epsilon$ if and only if  $d(a_1,A_2)\leq \epsilon$ for any $a_1\in A_1$, and $d(a_2,A_1)\leq \epsilon$ for any $a_2\in A_2$. In addition, if  $D(A_1,A_2)<\epsilon$, then for any $a_1\in A_1$ there exists $a_2\in A_2$ such that $d(a_1,a_2)<\epsilon$, and for any $b_2\in A_2$ there exists $b_1\in A_1$ such that $d(b_2,b_1)<\epsilon$.  These properties explain why the Hausdorff distance is the natural distance to be used when discussing approximation and stability in the context of sets:  suppose that our resolution is at most $r$, i.e., we are not able to distinguish between two points whose distance is at most some given positive number $r$. If it is known that $D(A_1,A_2)<r$, then we cannot distinguish between the sets $A_1$ and $A_2$, at least not by inspections based only on distance measurements. As a result of the above discussion, the intuitive phrase ``two sets have almost the same shape'' can be formulated precisely: the Hausdorff distance between the sets is smaller than some given positive parameter (note that a set and a rigid transformation of it usually have different shapes). 

\section{Exact formulation of the main question and informal formulation of the main result}\label{sec:FormalInformal}
The exact formulation of Question \ref{ques:main} is based on the concept of Hausdorff distance for reasons which were explained at the end of the previous section. 
\begin{question}
Suppose that $(P_k)_{k\in K}$ is a tuple of non-empty sets  in $X$. Let $(R_k)_{k\in K}$ be the corresponding Voronoi diagram.   Is it true that a small change of the sites yields a small change in the corresponding Voronoi cells, where both changes are measured with respect to the Hausdorff distance? More precisely, is it true that for any $\epsilon>0$ there exists $\Delta>0$ such that for any tuple $(P'_k)_{k\in K}$, the condition $D(P_k,P'_k)<\Delta$ for each $k\in K$ implies that $D(R_k,R'_k)<\epsilon$ for each $k\in K$, where $(R'_k)_{k\in K}$ is the Voronoi diagram of $(P'_k)_{k\in K}$? 
\end{question}
The main result (Theorem \ref{thm:stabilityUC}) says that the answer is positive. Here is an informal description of it:
\begin{answer}
Suppose that the underlying subset $X$ is a closed and convex set of a  (possibly infinite dimensional) uniformly convex normed space $\wt{X}$.  Suppose that a certain boundedness condition on the distance between points in $X$ and the sites holds, e.g., when $X$ is bounded  or when the sites form a (distorted) lattice. If there is a common positive lower bound on the distance between the sites, and the distance to each of them is attained, then indeed a small enough change of the (possibly infinitely many) sites yields a small change of the corresponding  Voronoi cells, where both changes are measured with respect to the Hausdorff distance; in other words, the shapes of the real cells and the corresponding perturbed ones are almost the same. Moreover, explicit bounds on the changes can be derived and they hold simultaneously for all the cells. There are counterexamples which show that the assumptions imposed above are crucial. 
\end{answer}
 The condition that the distance to a site is attained holds, e.g., when  the site is either a closed subset contained in a (translation of a) finite dimensional  space, or a compact set, or a convex and closed subset in a uniformly convex Banach space. The sites can always be assumed to be closed, since the distance and the Hausdorff distance preserve their values when the involved subsets are replaced by their closures. The ``certain boundedness condition on the distance between points in $X$ and the sites'' is a somewhat technical condition expressed in \eqref{eq:BallRhokThm} (see also Remark \ref{rem:rho}). 
 
\section{ The relevance of the main result}\label{sec:examples}
In Section \ref{sec:Intro} we explained why Question \ref{ques:main} is natural and fundamental, and mentioned the real-world example of a Voronoi diagram induced by shops/cellular antennas. The goal of this section is to illustrate further  the relevance of the main result using a (far from being exhaustive) list of real-world and theoretical exampls. In these examples the shape or the position of the real sites are obviously approximated, and the main result (Theorem \ref{thm:stabilityUC}) ensures that the   approximate Voronoi cells and the real ones have almost the same shape, and no unpleasant phenomenon such as the one described in Figures  \ref{fig:InStability000}-\ref{fig:InStability1Full} can occur.

One example is in molecular biology for modeling the proteins structure  
(Richards \cite{Richards}, Kim et al. \cite{KKCRCP}, Bourquard et al. \cite{VoronoiBiology2}), where the sites are either the atoms of a protein or special selected points in the amino acids and they are approximated by spheres/points. Another example is related to collision detection and robot motion (Goralski-Gold  \cite{GoralskiGold}, Schwartz et al. 
\cite{SchwartzSharirHopcroft}), where the sites are the (static  or dynamic) obstacles located in an environment in which a vehicle/airplane/ship/robot/satellite should move. A third example is in solid state physics (Ashcroft-Mermin \cite{AshcroftMermin}; here the common terms are ``the first Brillouin zone'' or ``the Wigner-Seitz cell'' instead of ``the Voronoi cell''), where the sites are infinitely many point  atoms in a (roughly)  periodic structure which represents a crystal. A fourth example is in material engineering (Li-Ghosh \cite{LiGhosh}), where the sites are cracks in a material. 

A fifth example is in numerical simulations of various dynamical phenomena, e.g., gas, fluid or urban dynamics (Slotterback et al. \cite{GranularMatter}, Mostafavi et al.  \cite{VoronoiSpatial}). Here the sites are certain points/shapes taken from the sampled data of the  simulated phenomena, and the cells help to cluster and analyze the data continuously. A sixth example is in astrophysics (Springel et al.  \cite{DarkMatterGalactic}) where the (point) sites are actually huge regions in the universe (of diameter equals to hundreds of light years) used in simulations performed for understanding the behavior of (dark) matter in the universe. A seventh example is in image processing and computer graphics,  where the sites are either certain important features/parts in an image (Tagare et al. \cite{TagareJaffeDuncan}, Dobashi et al. \cite{DobashiHagaJohanNishita},  Sabha-Dutr\'e \cite{SabhaDutre}) used for processing/analyzing it, or they are points form a useful configuration such as (an approximate) centroidal Voronoi diagram (CVD) which  induces cells  having good shapes (Du et al. \cite{VoronoiCVD_Review},  Liu et al. \cite{Graphics_CVD}, Faustino-Figueiredo \cite{FaustinoFigueiredo}). 

An eighth example is in computational geometry, and it is actually a large  collection of familiar problems in this field where Voronoi cells appear  
and being used, possibly indirectly: (approximate) nearest neighbor searching/the post office problem, cluster analysis, (approximate) closest pairs, motion planning, finding (approximate) minimum spanning trees, finding good triangulations, and so on. See, e.g., Aurenhammer \cite{Aurenhammer}, Aurenhammer-Klein  \cite{AurenhammerKlein}, Clarkson \cite{ClarksonNN}, Indyk \cite{Indyk}, and Okabe et al.  \cite{OBSC}.  Here the  sites are either points or other shapes, and the space is usually $\R^n$ with some norm. In some of the above problems  our stability result is clearly related because of the analysis being used (e.g., cluster analysis)  or because the position/shapes of the sites are approximated (e.g., motion planning, the post office problem). However, it  may be related also because  in many of the previous problems the difficulty/complexity of finding an exact  solution is too high, so one is forced to use approximate algorithms,  to impose a general position assumption, and so on. Now, by perturbing  slightly the sites (moving points, replacing a non-smooth curve by an approximating smooth one, etc.,) one may obtain much simpler and tractable configurations  and now, using the geometric stability of the Voronoi cells, one may estimate how well the obtained solution approximates the best/real one.  

As for a theoretical example of a different nature, we mention Kopeck\'a et al.  \cite{KopeckaReemReich} in which the stability results described here have been used, in particular,  for proving the existence of a zone diagram (a concept which was first introduced by Asano et al. \cite{AMTn} in the Euclidean plane with point sites) of finitely many compact sites which are  strictly contained in a (large) compact and convex subset of a uniformly convex space, and also for proving  interesting properties of Voronoi cells there. 

Another example  is for the infinite dimensional Hilbert space $L_2(I)$ for some $I$ (perhaps an interval or a finite dimensional space): functions in it are used in signal processing and in many other areas in science and technology. In practice the signals (functions) can be distorted, e.g.,  because of  noise, and in addition, they are approximated by finite series (e.g., finite  Fourier  series) or integrals (e.g., Fourier transform). Given several signals, the (approximate)  Voronoi cell of a given signal may help, at least in theory, to cluster or analyze data related to the sites. Such an analysis can be done also when the signal is considered  as a point in a finite dimensional space of a high dimension. See, for instance, Conway-Sloane \cite[pp. 66-69,  451-477]{ConwaySloane} (coding) and Shannon  \cite{Shannon} (communication) for a closely related discussion (in \cite{Shannon} Voronoi diagrams are  definitely used in various places, but without their explicit name).

We mention several additional examples  related to our stability result,  sometimes in a somewhat unexpected way. For instance, Voronoi diagrams of  infinitely many sites generated by a Poisson process (Okabe et al. \cite[pp. 39, 291-410]{OBSC}), Voronoi diagrams of atom nuclei used for the mathematical analysis of stability phenomena in matter (Lieb-Yau  \cite{LiebYau}),  Voronoi diagrams of infinitely many lattice points in a multi-dimensional Euclidean space which appear in the original  works of Dirichlet \cite{Dirichlet} and Voronoi \cite{Voronoi} (see also Groemer \cite{Groemer} and Gruber-Lekkerkerker \cite{GruberLek}  regarding the geometry of numbers and quadratic  forms; Groemer used his stability result for deriving the Mahler compactness theorem \cite{Mahler}),
and packing problems such as the Kepler conjecture and the Dodecahedral conjecture (Hales \cite{HalesKepler},\cite{KeplerDCG}, Hales-McLaughlin \cite{HalesMcLaughlin}; because of continuity arguments needed in the proof) or those described in Conway-Sloane \cite{ConwaySloane}.

\section{The main result and some aspects related to its proof}\label{sec:Outline}
In this section we formulate the main result and discuss briefly issues  related to its proof. See also the remarks after  
Theorem \ref{thm:stabilityUC} for several relevant clarifications. 

\begin{thm}\label{thm:stabilityUC}
Let $(\wt{X},|\cdot|)$ be a uniformly convex normed space. Let  $X\subseteq \wt{X}$ be closed and convex. Let $(P_k)_{k\in K}$, $(P'_k)_{k\in K}$ be two given tuples of nonempty subsets of  $X$ with the property that the distance between each $x\in X$ and each $P_k,P'_k$ is attained. For each $k\in K$ let $A_k=\bigcup_{j\neq k}P_j,\,A'_k=\bigcup_{j\neq k}P'_j$. Suppose that the following conditions hold:
\begin{equation}\label{eq:eta}
\eta:=\inf\{d(P_k,P_j): j,k\in K, j\neq k\}>0,
\end{equation}
%Suppose also that 
\begin{multline}\label{eq:BallRhokThm}
\exists \rho\in (0,\infty)\,\, \textnormal{such that for all}\,\,k\in K\,\,\textnormal{and for all}\,\, x\in X\,\,\\
\textnormal{the open ball}\,\, B(x,\rho)\,\,\textnormal{intersects} \,\,A_k.
\end{multline}
For each $k\in K$ let $R_k=\dom(P_k,A_k),R'_k=\dom(P'_k,A'_k)$ be, respectively, the Voronoi cells associated with the original site $P_k$ and the perturbed one $P'_k$. Then for each $\epsilon\in (0,\eta/6)$ there exists $\Delta>0$ such that if   $D(P_k,P'_k)<\Delta$ for each $k\in K$, then $D(R_k,R'_k)<\epsilon$ for each $k\in K$.
\end{thm}
See Figures  \ref{fig:Stability_0005_lpi_Before}, \ref{fig:Stability_0005_lpi_After} for an illustration. The pictures were produced using the algorithm described in \cite{ReemISVD09}. 

%%%%%%%%%%%%%%%%%%%%%%%%%%%%%%%%%%%%%%%%%%%%%%%%%%%%%%%%%%%%%%%%%%%%%%%
\begin{figure*}[t]
\begin{minipage}[t]{0.5\textwidth}
\begin{center}
{\includegraphics[scale=0.75]{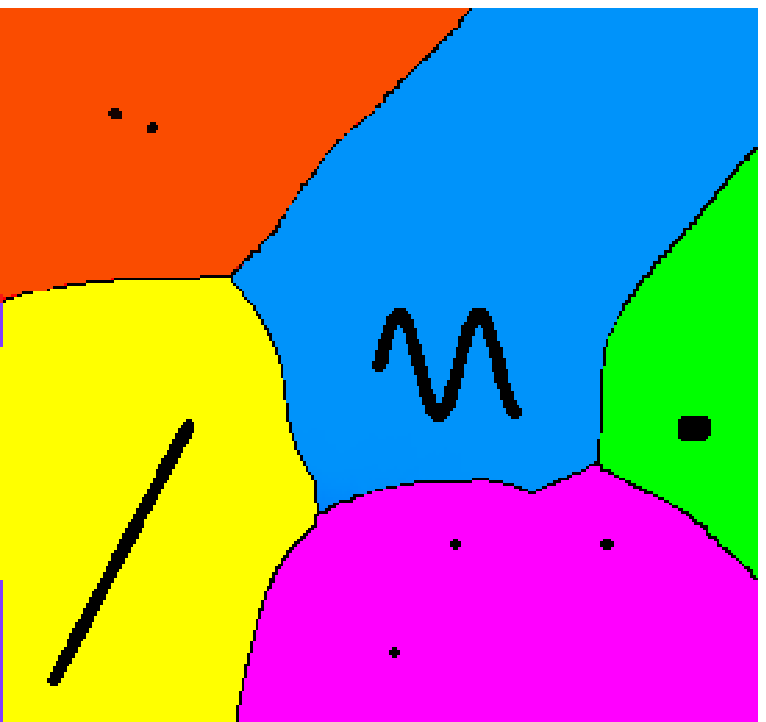}}
\end{center}
 \caption{Illustration of Theorem \ref{thm:stabilityUC}: five sites in a square in $(\R^2,\ell_p)$ where the parameter is $p=3.14159$.}
\label{fig:Stability_0005_lpi_Before} %\label{page:Stabilityl3}
\end{minipage}%
\hfill
\begin{minipage}[t]{0.5\textwidth}
\begin{center}
{\includegraphics[scale=0.75]{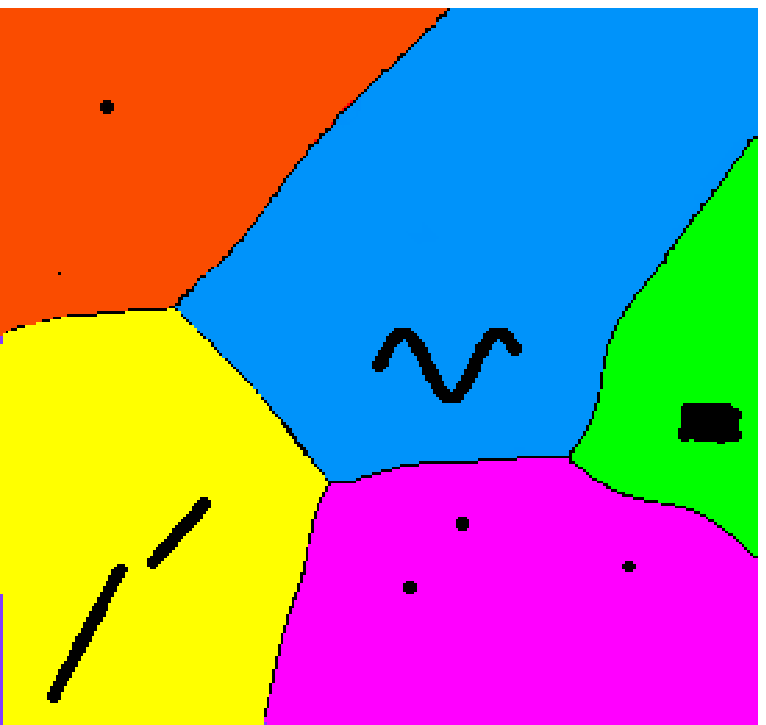}}
\end{center}
 \caption{The sites have been slightly perturbed: the two points have merged,  the ``sin'' has shrunk, and so on. The cells have been slightly perturbed.}
\label{fig:Stability_0005_lpi_After}
\end{minipage}%
\end{figure*}
%%%%%%%%%%%%%%%%%%%%%%%%%%%%%%%%%%%%%%%%%%%%%%%%%%%%%%%%%%%%%%%
\begin{remark}\label{rem:rho}
The assumption mentioned in \eqref{eq:BallRhokThm} may seem somewhat  complicated at a first glance, but it actually expresses a certain uniform  boundedness condition on the distance between any point in $X$ to  its  neighbor sites. No matter which point $x\in X$ and which site $P_k$ are chosen, the distance between $x$ and the collection of other sites $P_j,\,j\neq k$ cannot be arbitrary large. A sufficient condition for it to hold is when a uniform bound on the diameter of the cells (including the neutral one, if it is nonempty) is known in advance, e.g., when $X$ is bounded or when the sites form a (distorted) lattice. But \eqref{eq:BallRhokThm} can hold even if the cells are not bounded, e.g., when the setting is the Euclidean plane and $P_k=\R\times \{k\}$ where $k$ runs over all integers. 
\end{remark}
\begin{remark}\label{rem:Delta}
In general, we have $\Delta=O(\epsilon^2)$. However,  if there is a positive lower bound on the distance between  the sites and the boundary of $X$ (relative to the affine hull spanned by $X$), i.e., if  the sites are strictly contained in the (relative) interior of $X$, then actually the better estimate $\Delta=O(\epsilon)$ holds. The constants inside the big $O$ can be described explicitly: when $\Delta=C\epsilon^2$ we can take
\begin{equation*}%\label{eq:Delta}
C=\frac{1}{16(\rho+5\eta/12)}\cdot\delta\left(\frac{\eta}{12\rho+5\eta}
\right),
\end{equation*} 
and when $\Delta=C\epsilon$ we can take 
\begin{equation*}%\label{eq:Delta}
C=\min\left\{\displaystyle{ \frac{1}{16}\delta\left(\frac{\eta}{12\rho+5\eta}\right), \frac{d(\bigcup_{k\in K}P_k,\partial X)}{8(\rho+\eta/6)}}\right\}.
\end{equation*} \\
In the second case, in addition to $\epsilon <\eta/6$, the inequality 
$\epsilon\leq 8\cdot d(\bigcup_{k\in K}P_k,\partial X)$  should be satisfied too.
\end{remark} 
The proof of Theorem \ref{thm:stabilityUC} is quite long and technical, and hence it is given in the appendix. Despite this, we want to say a few words about the proof and some of the difficulties which arise along the way. First, as the counterexamples mentioned in Section \ref{sec:CounterExamples} show, one must take into account all the assumptions mentioned in the formulation of the theorem. 

Second, in order to arrive to the generality described in the theorem, 
one is forced to avoid many familiar arguments used in computational geometry and elsewhere, such as Euclidean arguments (standard angles, trigonometric functions, normals, etc.), arguments  based on lower envelopes and algebraic arguments (since the intersections between the graphs generating the lower envelope may be complicated and since the boundaries of the cells may not be algebraic), arguments based on continuous motion of points, arguments based on finite dimensional properties such as compactness (since in infinite dimensional spaces closed and bounded balls are not compact), arguments based on finiteness (since we allow infinitely many sites and sites consist of infinitely many points) and so on. Our way to overcome these difficulties is to use a new  representation theorem for dominance regions as a collection of line segments (Theorem \ref{thm:domInterval} below) and a new geometric lemma (Lemma \ref{lem:StrictSegment} below) which enables us to arrive to the explicit  bounds mentioned in the theorem. As a matter of fact, we are not aware of any other way to obtain these explicit bounds even in a square in the Euclidean plane with point sites. These tools also allow us to overcome the difficulty of a potential infinite accumulated error due to the possibility of  infinitely many sites/sites with infinitely many points/infinite dimension.  

\begin{thm}\label{thm:domInterval}
Let $X$ be a closed and convex subset of a normed space $(\widetilde{X},|\cdot|)$, and
 let $P,A\subseteq X$ be nonempty. Suppose that for all $x\in X$ the distance between $x$ and $P$ is attained. Then $\dom(P,A)$ is a union of line segments starting at the points of $P$. More precisely, given $p\in P$ and a unit vector $\theta$, let
\begin{multline*}%\label{eq:Tdef}
T(\theta,p)=\sup\{t\in [0,\infty): p+t\theta\in X\,\,\mathrm{and}\,\ 
 d(p+t\theta,p)\leq d(p+t\theta,A)\}.
\end{multline*}
Then
\begin{equation*}\label{eq:dom}
\dom(P,A)=\bigcup_{p\in P}\bigcup_{|\theta|=1}[p,p+T(\theta,p)\theta].
\end{equation*}
When $T(\theta,p)=\infty$, the notation $[p,p+T(\theta,p)\theta]$ means the ray $\{p+t\theta: t\in [0,\infty)\}$.
\end{thm}

%%%%%%%%%%%%%%%%%%%%%%%%%%%%%%%%%%%%%%%%%%%%%%%%%%%%%%%%%%%%%%%

\begin{lem}\label{lem:StrictSegment}
Let $(\widetilde{X},|\cdot|)$ be a uniformly convex normed space and let 
$A\subseteq \widetilde{X}$ be nonempty. Suppose that $y,p\in \widetilde{X}$ satisfy $d(y,p)\leq d(y,A)$ and $d(p,A)>0$. Let $x\in [p,y)$. Let $\sigma\in (0,\infty)$ be arbitrary. Then $d(x,p)<d(x,A)-r$ for any $r>0$ satisfying
\begin{equation*}%\label{eq:EpsilonUniform}
r\leq\!\min\!\left\{\sigma,\frac{4d(p,A)}{10},d(y,x)\delta\left(\frac{d(p,A)}{10(d(x,A)+\sigma+d(y,x))}\right)\!\right\}\!.
\end{equation*}
\end{lem}

The proof of Lemma \ref{lem:StrictSegment} is based on the  strong triangle inequality of Clarkson 
\cite[Theorem 3]{Clarkson}. It is interesting to note that although  this inequality was formulated in the very famous paper \cite{Clarkson} of Clarkson, it seems that it has been almost totally forgotten, and in fact, despite a comprehensive search we have made in the literature, we have found evidences to its existence only in \cite{Clarkson} and later in  \cite{Plant}.

\section{Counterexamples}\label{sec:CounterExamples}
%%%%%%%%%%%%%%%%%%%%%%%%%%%%%%%%%%%%%%%%%%%%%%%%%%%%%%%%%%%%%%%
\begin{figure*}
\begin{minipage}[t]{0.5\textwidth}
\begin{center}
{\includegraphics[scale=0.8]{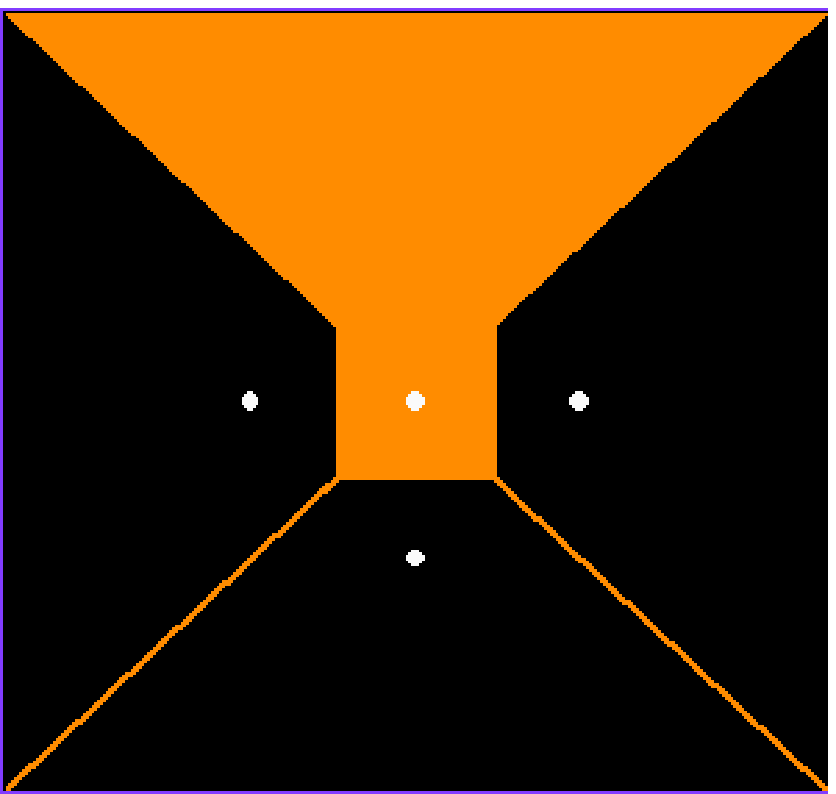}}
%{\includegraphics[scale=0.6]{InStability000.png}}%{Image16.png}}

\end{center}
 \caption{Four sites in a square in $(\R^2,\ell_{\infty})$. The cell of $P_1=\{(0,0)\}$ is displayed. The other sites 
 are $\,P_2=\{(2,0)\}$,$\,P_3=\{(-2,0)\}$, $\,P_4=\{(0,-2)\}$.}
\label{fig:InStability000} \label{page:InStability000}
\end{minipage}%
\hfill
\begin{minipage}[t]{0.5\textwidth}
\begin{center}
{\includegraphics[scale=0.8]{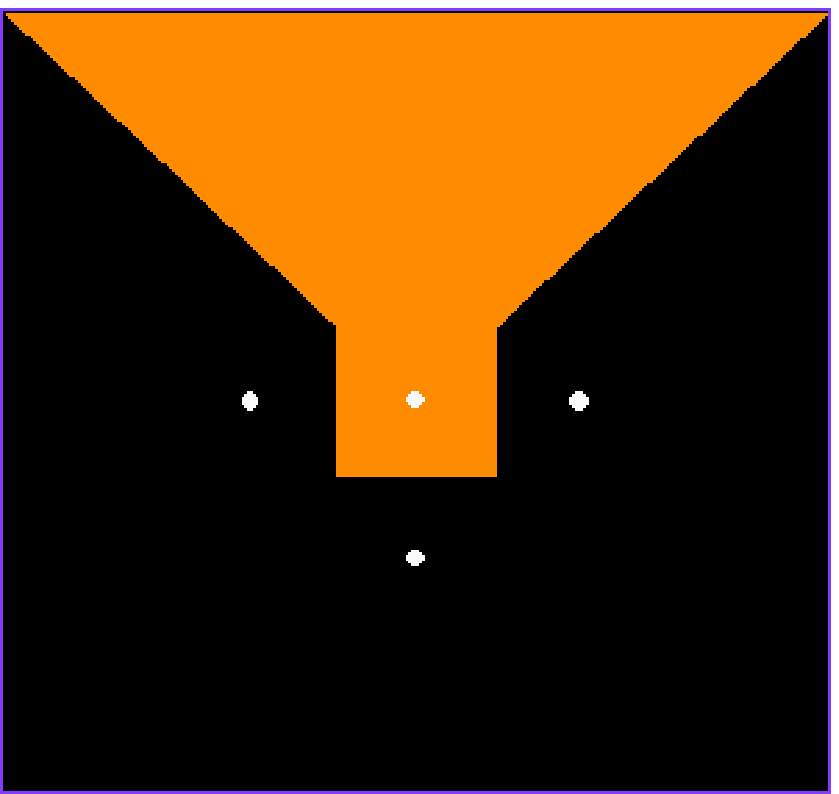}}
%{\includegraphics[scale=0.6]{InStability001.png}}%{Image20.png}}}
\end{center}
 \caption{Now either $P_4$ is the  square $[-\beta,\beta]\times [-2-\beta,-2+\beta]$  or $P_1=\{(\beta,\beta)\}$, $\beta>0$ arbitrary small. The two lower rays have disappeared. No stability.}
\label{fig:InStability001}
\end{minipage}%
\hfill
\begin{minipage}[t]{0.5\textwidth}
\begin{center}
{\includegraphics[scale=0.8]{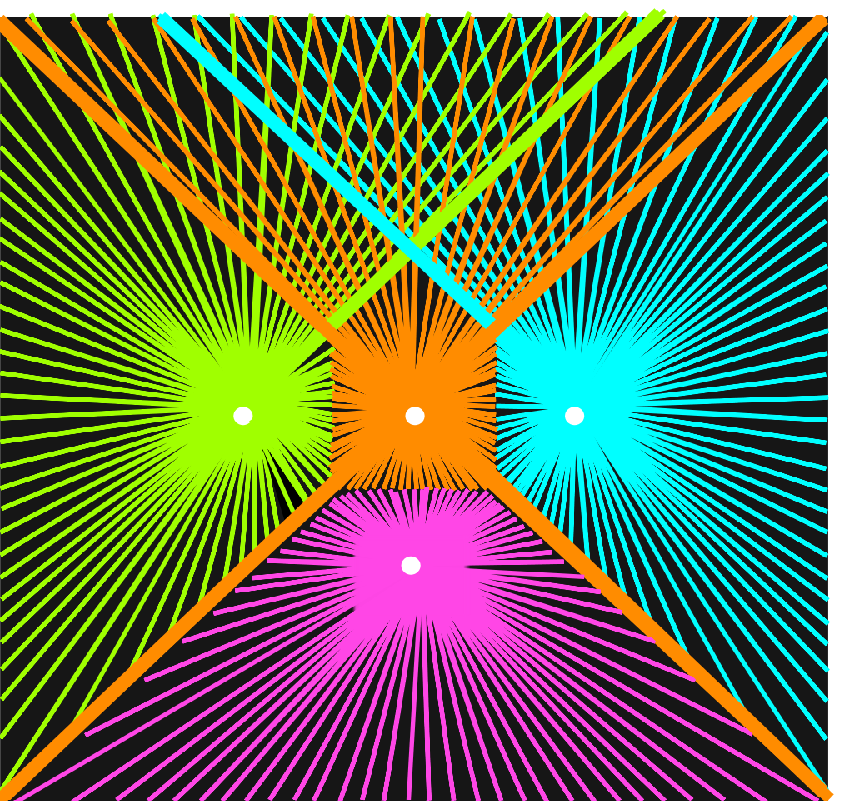}}%{Image16.png}
%{\includegraphics[scale=0.8]{InStability0Full.png}}%{Image16.png}}
\end{center}
 \caption{The full diagram of Figure \ref{fig:InStability000}. Note the large intersection between cells 1,2, and 3. For emphasizing this intersection, each cell  is represented as a union of rays (see Theorem~ \ref{thm:domInterval} for more information) and some rays were  emphasized.}
\label{fig:InStability0Full} \label{page:InStability000}
\end{minipage}%
\hfill
\begin{minipage}[t]{0.5\textwidth}
\begin{center}
{\includegraphics[scale=0.8]{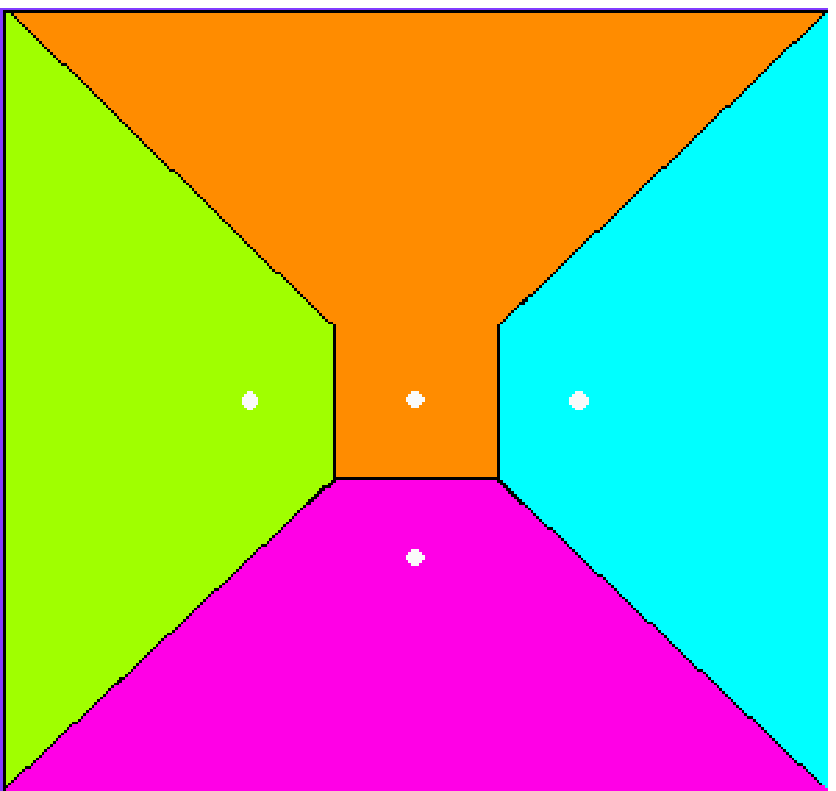}}
%{\includegraphics[scale=0.8]{InStability1Full.png}}%{Image20.png}}}
\end{center}
 \caption{The full diagram of Figure \ref{fig:InStability001} when $P_1=\{(\beta,\beta)\}$. Cells 2,3 have been (significantly) changed too.}
\label{fig:InStability1Full}
\end{minipage}%
\end{figure*}
%%%%%%%%%%%%%%%%%%%%%%%%%%%%%%%%%%%%%%%%%%%%%%%%%%%%%%%%%%%%%%%%%%%%%%%%%%%%

In this section we mention several counterexamples which show that the assumptions in  Theorem \ref{thm:stabilityUC} are essential. 

If the space is not uniformly convex, then the Voronoi cells may not be stable as shown in Figures  
\ref{fig:InStability000}-\ref{fig:InStability1Full}.  Here the setting is point sites in a square in $\R^2$ with the max norm. 

The positive common lower bound expressed in \eqref{eq:eta} is necessary even in a square in the Euclidean plane. Consider $X=[-10,10]^2$, $P_{1,\beta}=\{(0,\beta)\}$ and $P_{2,\beta}=\{(0,-\beta)\}$, where $\beta\in [0,1]$. As long as $\beta>0$, the cell  $\dom(P_{1,\beta},P_{2,\beta})$ is the upper half of $X$. However, if   $\beta=0$, then $\dom(P_{1,0},P_{2,0})$ is $X$. A more interesting example occurs when considering in $X$ the rectangle $P_{1,\beta}=[-a,a]\times [-10,-\beta]$ and the line segment  $P_2=[-10,10]\times \{0\}$, where $a,\beta\in [0,1]$.  If $\beta=0$, then $d(P_{1,0},P_2)=0$, and the cell $\dom(P_{1,0},P_2)$ contains the rectangle $[-a,a]\times [0,10]$. However, if  $\beta>0$, then this cell does not contain this rectangle. 

The assumption expressed in \eqref{eq:BallRhokThm} is essential even in the Euclidean plane with two points. Indeed, given $\beta\geq 0$, let $P_{1,\beta}=\{(\beta,1)\}$ and $P_2=\{(0,-1)\}$. Then $\dom(P_{1,0},P_2)$ is the upper half space. However, if $\beta>0$,  then the half space $\dom(P_{1,\beta},P_2)$ contains points $(x,y)$ with $y\to -\infty$. Thus the Hausdorff distance between the  original cell $\dom(P_{1,0},P_2)$ and the perturbed one  $\dom(P_{1,\beta},P_2)$ is  $\infty$, so there can be no stability.

\section{Concluding Remarks}\label{sec:Concluding}
We conclude the paper with the following remarks.
First, despite the counterexamples mentioned above, some of the assumptions can be weakened, with some caution. For instance, under a certain compactness assumption and a certain geometric  condition that the sites should satisfy,  the main result can be generalized to normed spaces which are not uniformly  convex. A particular case of these assumptions is when the space is  $m$-dimensional with the $|\cdot|_{\infty}$ norm, the sites are positively separated, and no two points of different sites are on a hyperplane perpendicular to one of the standard axes.  

Second, an interesting (but not immediate) consequence of the main result is that it implies the stability of the (multi-dimensional) volume, namely a small change in the sites yields a small change in the volumes of the corresponding Voronoi cells. 

Third, it can be shown that the function $T$ defined in Theorem~  \ref{thm:domInterval}  also has a certain continuity property if the space is uniformly convex and this expresses a certain continuity property of the boundary of the cells. 

Fourth, the estimate for $\Delta$ from Remark \ref{rem:Delta} is definitely not optimal and it can be improved, but, as simple examples show, $\Delta$ cannot be too much larger and its estimate should be taken into account when  performing a relevant analysis. There is nothing strange in this and the situation is analogous to the familiar case of real valued functions. For instance, consider  $f:\R\to\R$ defined by $f(x)=0$, $\,x<0$,  $f(x)=(1/\beta)x,\,x\in [0,\beta]$, $f(x)=1,\,\,x>\beta$, where $\beta>0$ is given. Although $f$ is continuous, a ``large''  change near $x=0$ (of more than $\beta$) will cause a large change to $f$.

\section*{Acknowledgments}
I want to thank Dan Halperin, Maarten L{\"o}ffler,  and Simeon Reich for helpful discussions regarding some of the references. I also thank all the  reviewers for their comments. 

\bibliographystyle{amsplain}
\bibliography{biblio}

\providecommand{\bysame}{\leavevmode\hbox to3em{\hrulefill}\thinspace}
\providecommand{\MR}{\relax\ifhmode\unskip\space\fi MR }
\providecommand{\MRhref}[2]{%
  \href{http://www.ams.org/mathscinet-getitem?mr=#1}{#2}
}
\providecommand{\href}[2]{#2}
\begin{thebibliography}{10}

\bibitem{AGGKKRS}
P.~K. Agarwal, J.~Gao, L.~Guibas, H.~Kaplan, V.~Koltun, N.~Rubin, and
  M.~Sharir, \emph{Kinetic stable {D}elaunay graphs}, Proceedings of the 26th
  annual symposium on Computational geometry (SoCG 2010), pp.~127--136.

\bibitem{AGMR}
G.~Albers, L.~J. Guibas, J.~S.~B. Mitchell, and T.~Roos, \emph{Voronoi diagrams
  of moving points}, Internat. J. Comput. Geom. Appl \textbf{8} (1998),
  365--380.

\bibitem{AryaMalamatosMount}
S.~Arya, T.~Malamatos, and D.~M. Mount, \emph{Space-efficient approximate
  {V}oronoi diagrams}, Proc. 34th ACM Symp. on Theory of Computing (STOC 2002),
  721--730.

\bibitem{AMTn}
T.~Asano, J.~Matou{\v{s}}ek, and T.~Tokuyama, \emph{Zone diagrams: Existence,
  uniqueness, and algorithmic challenge}, SIAM J. Comput. \textbf{37} (2007),
  no.~4, 1182--1198, a preliminary version in SODA 2007, pp. 756-765.

\bibitem{AshcroftMermin}
N.~W. Ashcroft and N.~D. Mermin, \emph{Solid state physics}, Holt, Rinehart and
  Winston, New York, 1976.

\bibitem{AttaliBoissonnatEdelsbrunner}
D.~Attali, J-D. Boissonnat, and H.~Edelsbrunner, \emph{Stability and
  computation of medial axes: a state of the art report}, Mathematical
  Foundations of Scientific Visualization, Computer Graphics, and Massive Data
  Exploration (B.~Hamann T.~M{\"o}ller and B.~Russell, eds.), Springer-Verlag,
  Mathematics and Visualization, 2007.

\bibitem{Aurenhammer}
F.~Aurenhammer, \emph{{V}oronoi diagrams - a survey of a fundamental geometric
  data structure}, ACM Computing Surveys, vol.~3, 1991, pp.~345--405.

\bibitem{AurenhammerKlein}
F.~Aurenhammer and R.~Klein, \emph{Voronoi diagrams}, Handbook of computational
  geometry (J. Sack and G. Urrutia, Eds.) (2000), 201--290.

\bibitem{BandyopadhyaySnoeyink}
D.~Bandyopadhyay and J.~Snoeyink, \emph{Almost-{D}elaunay simplices: nearest
  neighbor relations for imprecise points}, Proceedings of the 15th annual
  ACM-SIAM symposium on Discrete algorithms (SODA 2004), pp.~410--419.

\bibitem{VoronoiBiology2}
T.~Bourquard, J.~Bernauer, J.~Az\'e, and A.~Poupon, \emph{Comparing {V}oronoi
  and {L}aguerre tessellaions in the protein-protein docking context},
  Proceedings of the sixth international symposium on {V}oronoi diagrams in
  science and engineering (ISVD 2009), pp.~225--232.

\bibitem{ChazalCohenSteinerLieutier}
F.~Chazal, D.~Cohen-Steiner, and A.~Lieutier, \emph{A sampling theory for
  compact sets in {E}uclidean space}, Discrete Comput. Geom. \textbf{41}
  (2009), 461--479, A preliminary version in SoCG 2006, pp. 319-326.

\bibitem{ChoiSeidel}
S.-W. Choi and H-P Seidel, \emph{Linear {o}ne-sided stability of {MAT} for
  weakly injective {3D} domain}, Computer-Aided Design \textbf{36} (2004),
  no.~2, 95--109, A preliminary version in SMA 2002 pp. 344-355.

\bibitem{Clarkson}
J.~A. Clarkson, \emph{Uniformly convex spaces}, Trans. Amer. Math. Soc.
  \textbf{40} (1936), 396--414.

\bibitem{ClarksonNN}
K.~L. Clarkson, \emph{Nearest-neighbor searching and metric space dimensions},
  Nearest-Neighbor Methods for Learning and Vision: Theory and Practice
  (G.~Shakhnarovich, T.~Darrell, and P.~Indyk, eds.), MIT Press, 2006,
  pp.~15--59.

\bibitem{SteinerEdelsbrunerHarer}
D.~Cohen-Steiner, H.~Edelsbrunner, and J.~Harer, \emph{Stability of persistence
  diagrams}, Discrete Comput. Geom \textbf{37} (2007), no.~2, 103--120, A
  preliminary version in SoCG 2005 pp. 263-271.

\bibitem{ConwaySloane}
J.~H. Conway and N.~J.~A. Sloane, \emph{Sphere packings, lattices, and groups},
  third ed., Springer-Verlag, New York, 1999.

\bibitem{Dirichlet}
L.~Dirichlet, \emph{{\"U}ber die reduction der positiven quadratischen {F}ormen
  mit drei unbestimmten ganzen zahlen.}, J. Reine. Angew. Math. \textbf{40}
  (1850), 209--227.

\bibitem{DobashiHagaJohanNishita}
Y.~Dobashi, T.~Haga, H.~Johan, and T.~Nishita, \emph{A method for creating
  mosaic images using voronoi diagrams}, In Proceedings of Eurographics 2002
  Short Presentations (2002), 341--348.

\bibitem{VoronoiCVD_Review}
Q.~Du, V.~Faber, and M.~Gunzburger, \emph{Centroidal {V}oronoi tessellations:
  applications and algorithms}, SIAM Rev. \textbf{41} (1999), no.~4, 637--676
  (electronic).

\bibitem{FaustinoFigueiredo}
G.~M. Faustino and L.~H. Figueiredo, \emph{Simple adaptive mosaic effects},
  18th Brazilian Symposium on Computer Graphics and Image Processing (SIBGRAPI
  2005), 315--322.

\bibitem{FortuneStability}
S.~Fortune, \emph{Numerical stability of algorithms for {$2$}{D} {D}elaunay
  triangulations}, Internat. J. Comput. Geom. Appl. \textbf{5} (1995), no.~1-2,
  193--213, A preliminary version in SoCG 1992, pp. 83-92.

\bibitem{GoebelReich}
K.~Goebel and S.~Reich, \emph{Uniform convexity, hyperbolic geometry, and
  nonexpansive mappings}, Monographs and Textbooks in Pure and Applied
  Mathematics, vol.~83, Marcel Dekker Inc., New York, 1984.

\bibitem{VoronoiWeb}
C.~Gold, \emph{The {V}oronoi {W}eb {S}ite},
  \url{http://www.voronoi.com/wiki/index.php?title=Main_Page}.

\bibitem{GoralskiGold}
R.~I. Goralski and C.~M. Gold, \emph{The development of a dynamic {GIS} for
  maritime navigation safety}, ISPRS Workshop on Updating Geo-spatial Databases
  with Imagery and The 5th ISPRS Workshop on DMGISs, 2007, pp.~47--50.

\bibitem{Groemer}
H.~Groemer, \emph{Continuity properties of {V}oronoi domains}, Monatsh. Math.
  \textbf{75} (1971), no.~5, 423--431.

\bibitem{GruberLek}
P.~M. Gruber and C.~G. Lekkerkerker, \emph{Geometry of numbers}, second ed.,
  North Holland, 1987.

\bibitem{HalesKepler}
T.~C. Hales, \emph{A proof of the {K}epler conjecture}, Annals of Mathematics
  \textbf{162} (2005), 1065--1185, arXiv:math/9811078.

\bibitem{KeplerDCG}
\bysame, \emph{A special issue of the journal {DCG} devoted to the unabridged
  version of the proof of the {K}epler conjecture (6 papers)}, Discrete and
  Computational Geometry (DCG) \textbf{36} (2006), 1--269, one paper is a joint
  work with S. P. Ferguson.

\bibitem{HalesMcLaughlin}
T.~C. Hales and S.~McLaughlin, \emph{The dodecahedral conjecture}, Journal of
  the American Mathematical Society \textbf{23} (2010), 299--344,
  arXiv:math/9811079 (1998).

\bibitem{HarPeled}
S.~{Har-Peled}, \emph{A replacement for {V}oronoi diagrams of near linear
  size}, Proc. 42nd Annu. IEEE Sympos. Found. Comput. Sci. (FOCS 2001),
  pp.~94--103.

\bibitem{Indyk}
P.~Indyk, \emph{Nearest neighbors in high-dimensional spaces}, Chapter 39 in:
  Handbook of Discrete and Computational Geometry (J.E. Goodman and J.
  O'Rourke, Eds.) (2004).

\bibitem{Kaplan}
C.~Kaplan, \emph{{V}oronoi diagrams and ornamental design}, ISAMA 1999,
  pp.~277--283.

\bibitem{Khanban}
A.~A. Khanban, \emph{Basic algorithms of computational geometry with imprecise
  input}, Ph.D. thesis, Imperial College London, University of London, 2005.

\bibitem{KKCRCP}
D.-S. Kim, D.~Kim, Y.~Cho, J.~Ryu, C.-H. Cho, and J.~Y. Park,
  \emph{Visualization and analysis of protein structures using {E}uclidean
  {V}oronoi diagram of atoms}, Lecture notes in computer science \textbf{3482}
  (2005), 993--1002, Also appeared in ICCSA 2005.

\bibitem{KopeckaReemReich}
E.~Kopeck\'a, D.~Reem, and S.~Reich, \emph{Zone diagrams in compact subsets of
  uniformly convex spaces}, Israel Journal of Mathematics, accepted for
  publication, preliminary versions in arXiv 1002.3583 (2010), CCCG 2010, pp.
  17-20.

\bibitem{LiGhosh}
S.~Li. and S.~Ghosh, \emph{Extended {V}oronoi cell finite element model for
  multiple cohesive crack propagation in brittle materials}, Int. J. Numer.
  Meth. Eng \textbf{65} (2006), 1028--1067.

\bibitem{LiebYau}
E.~H. Lieb and H.-T. Yau, \emph{The stability and instability of relativistic
  matter}, Commun. Math. Phys. \textbf{118} (1988), 177--213.

\bibitem{LindenTzafriri}
J.~Lindenstrauss and L.~Tzafriri, \emph{Classical {B}anach spaces, {II}:
  {F}unction spaces}, Springer, Berlin, 1979.

\bibitem{Graphics_CVD}
Y.~Liu, W.~Wang, B.~L\'evy, F.~Sun, D.-M. Yan, L.~Lu, and C.~Yang, \emph{On
  centroidal {V}oronoi tessellation—energy smoothness and fast computation},
  ACM Transactions on Graphics (TOG) \textbf{28} (2009), Article 101 (17
  pages), Presented at SIGGRAPH 2010.

\bibitem{LofflerPhD}
M.~L{\"o}ffler, \emph{Data imprecision in computational geometry}, Ph.D.
  thesis, Utrecht University, 2009.

\bibitem{LofflerKreveld}
M.~L{\"o}ffler and M.~van Kreveld, \emph{Largest bounding box, smallest
  diameter, and related problems on imprecise points}, Computational Geometry:
  Theory and Applications \textbf{43} (2010), 419--433, a preliminary version
  in WADS 2007, LNCS, pp. 446-457,.

\bibitem{Mahler}
K.~Mahler, \emph{On lattice points in $n$-dimensional star bodies {I}.
  existence theorems}, Proc. Roy. Soc. Lond. A \textbf{187} (1946), 151--187.

\bibitem{VoronoiSpatial}
M.~A. Mostafavi, L.~H. Beni, and K.~Hins-Mallet, \emph{Representing dynamic
  spatial processes using {V}oronoi diagrams: recent developments}, Proceedings
  of the sixth international symposium on {V}oronoi diagrams in science and
  engineering (ISVD 2009), pp.~109--117.

\bibitem{OBSC}
A.~Okabe, B.~Boots, K.~Sugihara, and S.~N. Chiu, \emph{Spatial tessellations:
  concepts and applications of {V}oronoi diagrams}, second ed., Wiley Series in
  Probability and Statistics, John Wiley \& Sons Ltd., Chichester, 2000, with a
  foreword by D. G. Kendall.

\bibitem{Plant}
A.~T. Plant, \emph{The differentiability of nonlinear semigroups in uniformly
  convex spaces}, Israel J. Math. \textbf{38} (1981), no.~3, 257--268.

\bibitem{ReemISVD09}
D.~Reem, \emph{An algorithm for computing {V}oronoi diagrams of general
  generators in general normed spaces}, Proceedings of the sixth international
  symposium on {V}oronoi diagrams in science and engineering (ISVD 2009),
  pp.~144--152.

\bibitem{ReemPhD}
\bysame, \emph{Voronoi and zone diagrams}, Ph.D. thesis, The Technion, Haifa,
  2010.

\bibitem{Richards}
F.~M. Richards, \emph{The interpretation of protein structures: total volume,
  group volume distributions and packing density}, Journal of Molecular Biology
  \textbf{82} (1974), 1--14.

\bibitem{SabhaDutre}
M.~Sabha and P.~Dutr\'e, \emph{Feature-based texture synthesis and editing
  using {V}oronoi diagrams}, Proceedings of the sixth international symposium
  on {V}oronoi diagrams in science and engineering (ISVD 2009), 165--170.

\bibitem{GuibasSalesinStolfi}
D.~Salesin, J.~Stolfi, and L.~Guibas, \emph{Epsilon geometry: Building robust
  algorithms from imprecise computations}, Proceedings of the 5th annual
  symposium on Computational geometry (SoCG 1989), pp.~208--217.

\bibitem{SchwartzSharirHopcroft}
J.~T. Schwartz, M.~Sharir, and J.~E. Hopcroft (eds.), \emph{Planning, geometry,
  and complexity of robot motion}, Ablex Publishing Corporation, Norwood, New
  Jersey, 1987.

\bibitem{Shannon}
C.~Shannon, \emph{Communication in the presence of noise}, Proc. IRE
  \textbf{37} (1949), no.~1, 10--21, Republished at Proc. IEEE 86 (1998), no.
  2, pp. 447-457.

\bibitem{GranularMatter}
S.~Slotterback, M.~Toiya, L.~Goff, J.~F. Douglas, and W.~Losert,
  \emph{Correlation between particle motion and {V}oronoi-cell-shape
  fluctuations during the compaction of granular matter}, Phys. Rev. Lett.
  \textbf{101} (2008), 258001.

\bibitem{DarkMatterGalactic}
V.~Springel, S.~D.~M. White, C.~S. Frenk, J.~F. Navarro, A.~Jenkins,
  M.~Vogelsberger, J.~Wang, A.~Ludlow, and A.~Helmi, \emph{Prospects for
  detecting supersymmetric dark matter in the galactic halo}, Nature
  \textbf{456} (2008), 73--76.

\bibitem{SugiharaIriInagakiImai}
K.~Sugihara, M.~Iri, H.~Inagaki, and T.~Imai, \emph{Topology-oriented
  implementation - an approach to robust geometric algorithms}, Algorithmica
  \textbf{27} (2000), 5--20.

\bibitem{TagareJaffeDuncan}
H.~Tagare, C.~C. Jaffe, and J.~Duncan, \emph{Medical image databases - a
  content-based retrieval approach}, Journal of the American Medical
  Informatics Association \textbf{4} (1997), 184--198.

\bibitem{EconomyFacilityPNAS}
J.~Um, S.-W. Son, S.-I. Lee, H.~Jeong, and B.~J. Kim, \emph{Scaling laws
  between population and facility densities}, PNAS \textbf{106} (2009),
  14236--14240.

\bibitem{Voronoi}
G.~Voronoi, \emph{Nouvelles applications des parametres continus {\`a} la
  theorie des formes quadratiques.}, J. Reine. Angew. Math. \textbf{134}
  (1908), 198--287.

\bibitem{VGT}
M.~N. Vyalyi, E.~N. Gordeyev, and S.~P. Tarasov, \emph{The stability of the
  {V}oronoi diagram}, Computational Mathematics and Mathematical Physics
  \textbf{36} (1996), 405--414.

\bibitem{VoronoiVirus}
P.~Wang, M.~Gonzalez, C.~A. Hidalgo, and A.-L. Barab\'asi, \emph{Understanding
  the spreading patterns of mobile phone viruses}, Science \textbf{324} (2009),
  1071--1076, arXiv 0906.4567 (2009).

\bibitem{Weller}
F.~Weller, \emph{Stability of {V}oronoi neighborship under perturbations of the
  sites}, The 9th International Canadian Conference on Computational Geometry
  (CCCG 1997).

\end{thebibliography}
\newpage

\section{Appendix 1: proof of the main theorem}\label{sec:appendix}
This appendix is devoted for the proofs. In the sequel, unless otherwise stated, $(\widetilde{X},|\cdot|)$ is a normed space; $X$ is a nonempty closed and convex subset of $\wt{X}$; $P,P',A$, and $A'$ are nonempty subsets of $X$;  $(P_k)_{k\in K}$ and $(P'_k)_{k\in K}$ are two tuples of nonempty subsets of $X$ representing the sites and the perturbed ones respectively; it is assumed that the distance between any point $x\in X$ and each of the subsets $P,P',P_k,P'_k$ (for any  $k\in K$) is attained; for each $k\in K$, we let $A_k=\bigcup_{j\neq k}P_j$ and $A'_k=\bigcup_{j\neq k}P'_j$; unit vectors will usually be denoted by $\theta$ or $\phi$. 
 
The following simple lemma is needed for proving the representation theorem below it.
\begin{lem}\label{lem:segment}
Let $\emptyset\neq A\subseteq \widetilde{X}$. Suppose that $y,p\in \widetilde{X}$ satisfy $d(y,p)\leq d(y,A)$. Then
$d(x,p)\leq d(x,A)$ for any $x\in [p,y]$.
\end{lem}
%\vspace*{0.2cm}
\begin{proof}
Let $a\in A$. Since  $d(y,p)\leq d(y,a)$ and $x\in [p,y]$, we have
\begin{equation*}
d(y,x)+d(x,p)=d(y,p)\leq d(y,a) \leq d(y,x)+d(x,a).
\end{equation*}
Thus $d(x,p)\leq d(x,a)$ for each $a\in A$, so $d(x,p)\leq d(x,A)$.
\end{proof}
\begin{thm}{\bf(The representation theorem)}\label{thm:domIntervalApp}
 Suppose that for all $x\in X$ the distance between $x$ and $P$ is attained. Then $\dom(P,A)$ is a union of line segments starting at the points of $P$. More precisely, given $p\in P$ and a unit vector $\theta$, let
\begin{equation}\label{eq:Tdef}
T(\theta,p)=\sup\{t\in [0,\infty): p+t\theta\in X\,\,\mathrm{and}\,\,
 d(p+t\theta,p)\leq d(p+t\theta,A)\}.
\end{equation}
Then
\begin{equation}\label{eq:dom}
\dom(P,A)=\bigcup_{p\in P}\bigcup_{|\theta|=1}[p,p+T(\theta,p)\theta].
\end{equation}
When $T(\theta,p)=\infty$, the notation $[p,p+T(\theta,p)\theta]$ means the ray $\{p+t\theta: t\in [0,\infty)\}$.
\end{thm}
\vspace*{0.2cm}
\begin{proof}
Given $p\in P$ and $|\theta|=1$, let $T(\theta,p)$ be defined as in \eqref{eq:Tdef}.
Obviously the set which defines $T$ is nonempty and $T(\theta,p)\geq 0$.  If $T(\theta,p)=0$, then $[p,p+T(\theta,p)\theta]\subseteq \dom(P,A)$. Otherwise, the segment $[p,p+s\theta]$ is contained in
$\dom(P,A)$ for any $s\in [0,T(\theta,p))$ by the definition of $T(\theta,p)$ and Lemma \ref{lem:segment}.  If $T(\theta,p)=\infty$, then this simply means that the ray $\{p+t\theta: t\in [0,\infty)\}$   is contained in $\dom(P,A)$, so  $[p,p+T(\theta,p)\theta]\subseteq \dom(P,A)$ by our notation. Otherwise, $T(\theta,p)$ is finite. To see that $p+T(\theta,p)\theta\in \dom(P,A)$, we simply use the fact that the function $t\mapsto d(p+t\theta,p)-d(p+t\theta,A)$ is non-positive for all $t\in [0,T(\theta,p))$, and that it is continuous. All the points are in $X$  because it is closed and convex. Thus $\bigcup_{p\in P}\bigcup_{|\theta|=1}[p,p+T(\theta,p)\theta]\subseteq \dom(P,A)$.

Now let $y\in \dom(P,A)$. By assumption, $d(y,P)=d(y,p)$ for some $p\in P$. If $y=p$, then $y$ is in the union $\bigcup_{q\in P}\bigcup_{|\theta|=1}[q,q+T(\theta,q)\theta]$. Otherwise, let $t=d(y,p)$ and
$\theta_0=(y-p)/t$. Then $\theta_0\in \{\theta: |\theta|=1\}$. In order to show that $y$ is in the union $\bigcup_{q\in P}\bigcup_{|\theta|=1}[q,q+T(\theta,q)\theta]$, it suffices to show that $y\in [p,p+T(\theta_0,p)\theta_0]$. Since $d(y,p)\leq d(y,A)$ and $y=p+t\theta\in X$, this is a simple consequence of  the definition of $T(\theta_0,p)$. 
\end{proof}

The following elementary lemma will be useful later. 
\begin{lem}\label{lem:InX}
Let $\wt{X}$ be a normed space. 
\begin{enumerate}[(a)]
\item\label{item:theta} Given $x,y,z\in \wt{X}$, $x\neq y, z\neq y$, 
\begin{equation}
\left|\frac{z-y}{|z-y|}-\frac{x-y}{|x-y|}\right|\leq \min\left\{2\frac{|z-x|}{|x-y|},2\frac{|z-x|}{|z-y|}\right\}.
\end{equation}
\item\label{item:B_H(x,r)X} Let $X\subseteq \wt{X}$ be nonempty and convex.  Let $p$ and $y$, $p\neq y$ be two points in $X$. Suppose that the open ball $B_H(p,R)=B(p,R)\cap H$, relative to the affine hull $H$ spanned by $X$, is contained in $X$. If $x\in (p,y)$, then $B_H(x,r)\subset X$ for any $r\leq \min\{d(x,y),d(x,p),Rd(x,y)/(2d(p,y))\}$.  
\end{enumerate}
\end{lem}
\begin{proof}
\begin{enumerate}[(a)]
\item From the triangle inequality,
\begin{multline*}
\left|\frac{z-y}{|z-y|}-\frac{x-y}{|x-y|}\right|=\left|\frac{|x-y|(z-y)-|z-y|(x-y)}{|z-y||x-y|}\right|=\\
\left|\frac{(|x-y|-|z-y|)(z-y)+|z-y|((z-y)-(x-y))}{|z-y||x-y|}\right|
\leq \frac{2|x-z||z-y|}{|z-y||x-y|}=\frac{2|x-z|}{|x-y|}.
\end{multline*}
The inequality $|(z-y)/|z-y|-(x-y)/|x-y||\leq 2|z-x|/|z-y|$ is proved in a similar way.
\item Let $z\in B_H(x,r)$. If $z=x$, then obviously $z\in X$. Now assume that $z\neq x$. By assumption $d(y,x)\geq r$ and $z\in B(x,r)$. Thus $y\neq z$. 
Let $\theta_1=(x-y)/|x-y|=(p-y)/|p-y|$ and $\theta_2=(z-y)/|z-y|$. The point $y+d(p,y)\theta_2$ is in $H$ since $[y,z]\subset H$. This point is also in $B(p,R)$ since 
\begin{multline*}
|(y+d(p,y)\theta_2)-p|=|(y+d(p,y)\theta_2)-(y+d(p,y)\theta_1)|
=d(p,y)|\theta_1-\theta_2|\\
\leq 2|p-y||z-x|/|x-y|<2|p-y|r/|x-y|\leq R
\end{multline*}
 by part \eqref{item:theta} and the choice of $r$. Therefore $y+d(p,y)\theta_2\in B_H(p,R)\subseteq X$. In addition, 
\begin{equation*}
d(z,y)\leq d(z,x)+d(x,y)<r+d(x,y)\leq d(p,x)+d(x,y)=d(p,y),
\end{equation*}
 and hence $z\in [y,y+d(p,y)\theta_2]\subseteq X$. 
\end{enumerate}
\end{proof}

The following technical proposition is central.
\begin{prop}\label{prop:stability}
Suppose that $d(P,A)>0$. Let $\epsilon>0$ be such that $\epsilon\leq d(P,A)/6$ and suppose that the following conditions hold:
\begin{multline}\label{eq:property}
\textnormal{There exists}\,\, \lambda\in (0,\epsilon)\,\, \textnormal{such that for each}\,\, p\in P, y\in \dom(P,A),\,x\in [p,y],\, \\
\textnormal{if}\,\, d(x,y)=\epsilon/2,\,\, \textnormal{then}\,\,d(x,p)\leq d(x,A)-\lambda.
\end{multline}
\begin{multline}\label{eq:property'}
\textnormal{There exists}\,\, \lambda'\in (0,\epsilon)\,\, \textnormal{such that for each}\,\, p'\in P',y'\in \dom(P',A'),\,x'\in [p',y']\, \\
\textnormal{if}\,\, d(x',y')=\epsilon/2,\,\, \textnormal{then}\,\,d(x',p')\leq d(x',A')-\lambda'.
\end{multline}
 Suppose that there are two positive numbers $M,M'$ such that
\begin{equation}\label{eq:MM'}
\sup\{T(\theta,p): p\in P,|\theta|=1\}\leq M,\quad  \sup\{T'(\theta',p'): p'\in P',|\theta'|=1\}\leq M',
\end{equation}
where $T'(\theta',p')$ is defined as in \eqref{eq:Tdef}, but with $A'$ instead of $A$. Let $\epsilon_1,  \epsilon_3$ be positive numbers satisfying
\begin{equation}\label{eq:epsilon_k}
\epsilon_1+\epsilon_3(3+4M/\epsilon)<\lambda/2, \quad
\epsilon_1+\epsilon_3(3+4M'/\epsilon)<\lambda'/2.
\end{equation}
If
\begin{equation}\label{eq:HausdorffEpsilon}
D(A,A')< \epsilon_1, \quad
D(P,P')< \epsilon_3,
\end{equation}
 then $D(\dom(P,A),\dom(P',A'))<\epsilon$.
\end{prop}
\begin{proof}
We first note that $5\epsilon\leq d(P',A')$. Indeed, let $p'\in P'$ and $a'\in A'$ be given.
By  \eqref{eq:HausdorffEpsilon} and \eqref{eq:epsilon_k} there are $p\in P$ and $a\in A$ such that $d(p,p')<\epsilon_3<\epsilon/2$ and $d(a',a)<\epsilon_1<\epsilon/2$, so
\begin{equation}\label{eq:d(P',A')}
6\epsilon\leq d(p,a)\leq d(p,p')+d(p',a')+d(a',a)< \epsilon+d(p',a'),
\end{equation}
which implies that $5\epsilon\leq d(P',A')$. We will show that there exists $\wt{\epsilon}\in (0,\epsilon)$ such that $D(\dom(P,A),\dom(P',A'))\leq \wt{\epsilon}<\epsilon$. In fact, we will show that we can take $\wt{\epsilon}:=\epsilon-\epsilon_3$. In order to prove that $D(\dom(P,A),\dom(P',A'))\leq \wt{\epsilon}$, it suffices to show that  for any $y\in \dom(P,A)$ there exists $y'\in \dom(P',A')$ such that $d(y,y')<\wt{\epsilon}$, and that for any $z'\in \dom(P',A')$ there exists $z\in \dom(P,A)$ such that $d(z',z)<\wt{\epsilon}$. We will show the first inequality, since the proof of the second one is  similar using \eqref{eq:property'},\eqref{eq:MM'}, and \eqref{eq:epsilon_k} with $\lambda'$ and $M'$. 

Let $y\in \dom(P,A)$. By Theorem \ref{thm:domIntervalApp} we  can write $y=p+t\theta$ for some $p\in P$, a unit vector  $\theta$, and
$t\in [0,T(\theta,p)]$. By \eqref{eq:MM'} we have $t\leq M$. By \eqref{eq:HausdorffEpsilon}  there is $p'\in P'$ such that  $d(p,p')<\epsilon_3$. Suppose first that $t\leq \epsilon$. 
Let $y'=y$. Then $y'\in X$, $d(y,y')<\epsilon-\epsilon_3$, and $d(y',p')\leq d(y',p)+d(p,p')< 2\epsilon$. In addition,
\begin{equation}\label{eq:4epsilonP'A'}
5\epsilon\leq d(P',A')\leq d(p',A')\leq d(p',y')+d(y',A')< 2\epsilon+d(y',A').
\end{equation}
Hence $d(y',p')< 2\epsilon< d(y',A')$ and $y'\in\dom(P',A')$.

Now assume that $\epsilon<t$.  Then $t=d(p,y)\leq d(p,p')+d(p',y)< \epsilon_3+|y-p'|$. Hence $0<t-\epsilon_3<|y-p'|$ since $t-\epsilon_3>t-\epsilon/2\geq \epsilon/2>0$.  Let $\theta'=(y-p')/|y-p'|$ and denote $\epsilon_2=2\epsilon_3/\epsilon$. Then, using Lemma \ref{lem:InX}\eqref{item:theta},  
\begin{equation}\label{eq:ThetaCont}
\left|\theta-\theta'\right|=\left|\frac{y-p}{|y-p|}-\frac{y-p'}{|y-p'|}\right|
\leq \frac{2|p-p'|}{|y-p|}<\frac{2\epsilon_3}{\epsilon}=\epsilon_2.
\end{equation}
Let $x=p+(t-0.5\epsilon)\theta$. Then $x\in [p,y]$ and $d(x,y)=\epsilon/2$, so by   \eqref{eq:property} we have $d(x,p)\leq d(x,A)-\lambda$. 
Let $y'=p'+(t-0.5\epsilon)\theta'$. Because $t-\epsilon/2<|y-p'|$ (see the discussion above \eqref{eq:ThetaCont}) we have  $y'\in [p',y]\subseteq X$. In addition, 
\begin{multline*}
d(y,y')\leq d(y,x)+d(x,y')
\leq\epsilon/2+(t-0.5\epsilon)|\theta-\theta'|+d(p,p')
<\epsilon/2+M\epsilon_2+\epsilon_3<\epsilon-\epsilon_3
\end{multline*}
by \eqref{eq:MM'} and \eqref{eq:epsilon_k}. Recalling that $d(y',A)\leq d(y',A')+D(A,A')$ holds in general, %(see the facts after Definition ~\ref{def:Hausdorff}), 
we have 
\begin{multline*}%\label{eq:epsilon123}
d(y',p')\leq d(y',x)+d(x,p)+d(p,p')
<M\epsilon_2+\epsilon_3+d(x,A)-\lambda+\epsilon_3\\
\leq M\epsilon_2+\epsilon_3+d(x,y')+d(y',A')+D(A,A')
+\epsilon_3-\lambda\\
< 2M\epsilon_2+3\epsilon_3+\epsilon_1+d(y',A')-\lambda< d(y',A')-\lambda/2,
\end{multline*}
 where we used \eqref{eq:epsilon_k} and $\epsilon_2=2\epsilon_3/\epsilon$ in the last inequality. Thus $d(y',p')<d(y',A')$. Consequently, in  both cases we have $d(y,y')<\wt{\epsilon}:=\epsilon-\epsilon_3$ and $y'\in \dom(P',A')$, as  required. 
\end{proof}

\begin{prop}\label{prop:stabilityVor}
Suppose that $\inf\{d(P_k,P_j): j,k\in K,\,j\neq k\}>0$. Let $\epsilon>0$ be such that $\epsilon\leq\inf\{d(P_k,P_j): j,k\in K,\,j\neq k\}/6$. Suppose  that the following conditions hold:
\begin{multline}\label{eq:property_k}
\exists\,\lambda\in (0,\epsilon)\,\, \textnormal{such that for each}\,\, k\in K, p\in P_k,\,\,y\in \dom(P_k,A_k),\, \textnormal{and}\, x\in [p,y]\, \\
\textnormal{if}\,\,d(x,y)=\epsilon/2,\,\, \textnormal{then}\,\,d(x,p)\leq d(x,A_k)-\lambda.
\end{multline}
\begin{multline}\label{eq:property'_k}
\exists\,\lambda'\in (0,\epsilon)\,\, \textnormal{such that for each}\,\, k\in K, p'\in P'_k,\,y'\in \dom(P'_k,A'_k),\,\, \textnormal{and}\,\,x'\in [p',y'] \\
\textnormal{if}\,\, d(x',y')=\epsilon/2,\,\, \textnormal{then}\,\,d(x',p')\leq d(x',A'_k)-\lambda'.
\end{multline}
Let
\begin{equation*}\label{eq:R_k}
R_k=\dom(P_k,A_k), \quad
\quad R'_k=\dom(P'_k,A'_k).
\end{equation*}
Suppose that there are  $M,M'\in (0,\infty)$ such that for all $k\in K$,
\begin{equation}\label{eq:Mk}
\sup\{T_k(\theta,p): p\in P, |\theta|=1\}\leq M,\quad\sup\{T'_k(\theta',p'): p'\in P',|\theta'|=1\}\leq M',
\end{equation}
where $T_k(\theta,p)$ and $T'_k(\theta',p')$ are defined as in \eqref{eq:Tdef}, but with $A_k$ and $A'_k$ instead of $A$.
Let  $\epsilon_4$ be a positive number satisfying
\begin{equation}\label{eq:epsilon24M}
4(1+M/\epsilon)\epsilon_4<\lambda/2, \quad 4(1+M'/\epsilon)\epsilon_4<\lambda'/2.
\end{equation}
If
\begin{equation}\label{eq:epsilon24k}
  \quad D(P_k,P'_k)< \epsilon_4 \,\,\,\,\forall k\in K,
\end{equation}
 then $D(R_k,R'_k)<\epsilon$ for each $k\in K$.
\end{prop}
\begin{proof}
From the condition $\epsilon\leq d(P_k,P_j)/6$ for each $k\neq j$ it follows that $\epsilon\leq d(P_k,\bigcup_{j\neq k}P_j)/6$. In addition, the assumption $D(P_k,P'_k)<\epsilon_4$ for all $k\in K$ implies  $D(\bigcup_{i\in I}P_i,\bigcup_{i\in I}P'_i)\leq\epsilon_4$ for any $I\subseteq K$, and in particular this is true for $I=K\backslash\{k\}$, i.e.,  $D(A_k,A'_k)\leq \epsilon_4$. Indeed, we only have to observe that given $y\in \bigcup_{i\in I} P_i$, there exists $j\in I$ such that $y\in P_j$, so
\begin{equation*}
d(y,\bigcup_{i\in I}P'_i)\leq d(y,P'_j)\leq D(P_j,P'_j)<\epsilon_4,
\end{equation*}
and the inequality follows. Let $\epsilon_3:=\epsilon_4$ and  $\epsilon_1:=\epsilon_4+r$, where $r$ is defined as the minimum of  the two  positive numbers  $(\lambda/2-4\epsilon_4(1+M/\epsilon))/4$ and $(\lambda'/2-4\epsilon_4(1+M'/\epsilon))/4$. Then    $\epsilon_1+\epsilon_3(3+4M/\epsilon)<\lambda/2$ and $\epsilon_1+\epsilon_3(3+4M'/\epsilon)<\lambda'/2$.
In addition $D(P_k,P'_k)<\epsilon_3$ and  $D(A_k,A'_k)\leq\epsilon_4<\epsilon_1$ for each $k\in K$. Thus all the  conditions of Proposition \ref{prop:stability} are satisfied for each $k\in K$ with $P=P_k$ and $A=A_k$. Hence $D(R_k,R'_k)<\epsilon$.
\end{proof}

The following lemma provides a simple sufficient condition for \eqref{eq:MM'} to hold.
\begin{lem}\label{lem:BallRho}
Suppose that the following condition holds:
\begin{equation}\label{eq:BallRho}
\exists \rho\in (0,\infty)\,\, \textnormal{such that}\,\,\forall x\in X\,\,\textnormal{the open ball}\,\, B(x,\rho)\,\,\textnormal{intersects} \,\,A.
\end{equation}
Then
\begin{enumerate}[(a)]
\item\label{item:Trho} $\sup\{T(\theta,p): |\theta|=1, p\in P\}\leq \rho$.
\item\label{item:RhoAA'} $d(x,A)< \rho$ and $d(x,A')<\rho+D(A,A')$ for each $x\in X$.
\item\label{item:T'rho} $\sup\{T'(\theta',p'): |\theta'|=1, p'\in P'\}\leq \rho+D(A,A')$.
\end{enumerate}
In the same way, if $A_k=\bigcup_{j\neq k}P_j$ for each $k\in K$ and the following condition holds:
\begin{multline}\label{eq:BallRhok}
\exists \rho\in (0,\infty)\,\, \textnormal{such that for all}\,\,k\in K\,\,\textnormal{and for all}\,\, x\in X\,\,\\
\textnormal{the open ball}\,\, B(x,\rho)\,\,\textnormal{intersects} \,\,A_k,
\end{multline}
then all the above claims remain true with $T_k,P_k,A_k$ instead of  $T,P,A$.
\end{lem}
\begin{proof}
\begin{enumerate}[(a)]
\item Let $p\in P$ be arbitrary and let $\theta$ be an arbitrary unit vector.  Let $t=\rho$ and $x=p+t\theta$. If $x\notin X$ then $T(\theta,p)\leq t$, since otherwise there is some $s>t$ such that $p+s\theta\in X$ and $d(p+s\theta,p)\leq d(p+s\theta,A)$ by the definition of $T$ (see \eqref{eq:Tdef}). Hence $x\in [p,p+s\theta]\subseteq X$, a contradiction.  
Now assume that $x\in X$. By \eqref{eq:BallRho} there exists $a\in A$ such that $d(x,a)<\rho=d(x,p)$, so $d(x,A)<d(x,p)$. By the definition of $T$  and Lemma \ref{lem:segment}  we conclude that $T(\theta,p)\leq \rho$ in this case too.
\item Let $x\in X$. By \eqref{eq:BallRho} we know that $d(x,A)<\rho$. Since $d(x,A')\leq d(x,A)+D(A,A')$ holds in general, we have $d(x,A')<\rho+D(A,A')$.
\item If $D(A,A')=\infty$, then the assertion holds trivially. Assume now that $D(A,A')<\infty$. By part \eqref{item:RhoAA'} the open ball $B(x,\rho+D(A,A'))$ intersects $A'$ for each $x\in X$ and the proof continues as in part \eqref{item:Trho} with $\rho+D(A,A')$ instead of $\rho$.
\end{enumerate}
The proofs in the case where \eqref{eq:BallRhok} holds are the same as above.
\end{proof}

The next lemma is the key step in establishing \eqref{eq:property},\eqref{eq:property'}, \eqref{eq:property_k}, and \eqref{eq:property'_k} in the case of uniformly convex normed spaces. It improves upon Lemma \ref{lem:segment} under the additional assumptions that $(\widetilde{X},|\cdot|)$ is uniformly convex and $d(p,A)>0$. Its proof is based on the following definition and on a special case of the forgotten strong triangle inequality of Clarkson \cite[Theorem~3]{Clarkson}. See page \pageref{page:UniConvDef} for the definition of uniformly convex spaces and the meaning of the function $\delta$. Clarkson's theorem, which is formulated in \cite{Clarkson} for uniformly convex Banach spaces, has nothing to do with completeness and it remains true without this assumption. 
\begin{defin}\label{def:angle}
Given two non-vanishing vectors $x,y\in \wt{X}$, the angle (or Clarkson's angle, or the normed angle) $\alpha(x,y)$ between them is the distance between their directions, i.e., it is defined by
\begin{equation*}
\alpha(x,y)=\left|\frac{x}{|x|}-\frac{y}{|y|}\right|.
\end{equation*}
\end{defin}
\begin{thm}\label{thm:Clarkson}{\bf (Clarkson)}
Let $x_1,x_2$ be two non-vanishing vectors in a uniformly convex normed space $(\widetilde{X},|\cdot|)$. If $x_1+x_2\neq 0$, then
\begin{equation}\label{eq:ClarkTriangle}
|x_1+x_2|\leq |x_1|+|x_2|-2\delta(\alpha_1)|x_1|-2\delta(\alpha_2)|x_2|,
\end{equation}
where $\alpha_l=\alpha(x_l,x_1+x_2),\,l=1,2$.
\end{thm}

\begin{lem}\label{lem:StrictSegment_Appendix}
Let $(\widetilde{X},|\cdot|)$ be a uniformly convex normed space and let $\emptyset\neq A\subseteq \widetilde{X}$. Suppose that $y,p\in \widetilde{X}$ satisfy $d(y,p)\leq d(y,A)$ and $d(p,A)>0$. Let $x\in [p,y)$. Let $\sigma\in (0,\infty)$ be arbitrary. Then $d(x,p)<d(x,A)-r$ for any $r>0$ satisfying
\begin{equation}\label{eq:EpsilonUniform}
r\leq\min\left\{\sigma,\frac{4d(p,A)}{10},d(y,x)\delta\left(\frac{d(p,A)}{10(d(x,A)+\sigma+d(y,x))}\right)\right\}.
\end{equation}
\end{lem}
\begin{proof}
The assertion is obvious if either $p=y$ or $p=x$, so assume $p\neq y$ and $p\neq x$ from now on. Let $a\in A$ satisfy $d(x,a)<d(x,A)+r$. We distinguish between two cases. Suppose first that the angle $\alpha(x-y,a-y)$ is greater or equal to $d(p,A)/(10(d(x,A)+\sigma+d(y,x)))$. By the assumption, the strong triangle inequality \eqref{eq:ClarkTriangle} and the fact that $\delta$ is increasing,
\begin{multline*}
d(y,x)+d(x,p)=d(y,p)\leq d(y,A)\leq d(y,a)\\
\leq d(y,x)+d(x,a)-2\delta(\alpha(x-y,a-y))d(y,x)-2\delta(\alpha(a-x,a-y))d(x,a)\\
<d(y,x)+d(x,A)+r-2r-0.
\end{multline*}
Hence $d(x,p)+r<d(x,A)$ as required. We note that all the angles are well defined, since $0<d(x,p)\leq d(x,a)$ by Lemma \ref{lem:segment}, and $0<d(y,p)\leq d(y,a), \,0<d(y,x)$ by assumption. In addition, the expression in \eqref{eq:EpsilonUniform} is well defined, since $d(p,A)\leq d(p,x)+d(x,A)\leq 2d(x,A)$ by Lemma \ref{lem:segment}, so the argument inside $\delta$ is in the interval $[0,0.2]\subseteq [0,2]$. In addition, the  minimum in \eqref{eq:EpsilonUniform} is positive since all the numbers inside $\delta$ are positive and the space is uniformly convex. 

Assume now that $\alpha(x-y,a-y)< d(p,A)/(10(d(x,A)+\sigma+d(y,x)))$. Since $x$ is between $p$ and $y$, the unit vector $\theta=(p-y)/|p-y|$ satisfies  $p=y+d(y,p)\theta$ and $x=y+d(y,x)\theta$. Let $\phi=(a-y)/|a-y|$. Let $z=y+d(y,a)\theta$. Since $d(y,p)\leq d(y,a)$, it follows that $p$ is between $z$ and $x$ (and maybe $p=z$). Because $\alpha(x-y,a-y)=d(\theta,\phi)$ and $r\leq \sigma$, it follows from the assumption on $\alpha(x-y,a-y)$ that
\begin{equation*}
d(z,a)=d(\theta,\phi)d(y,a)<\frac{d(p,A)(d(y,x)+d(x,a))}{10(d(x,A)+\sigma+d(y,x))}<\frac{d(p,A)}{10}.
\end{equation*}
Thus
\begin{equation*}
d(p,A)\leq d(p,a)\leq d(p,z)+d(z,a)<d(p,z)+d(p,A)/10,
\end{equation*}
so $9d(p,A)/10<d(p,z)$. Since $p$ is between $z$ and $x$, the choice of $r$ in \eqref{eq:EpsilonUniform} implies
\begin{multline*}
d(x,A)+4d(p,A)/10\geq d(x,A)+r>d(x,a)\geq   d(x,z)-d(a,z)\\
=d(x,p)+d(p,z)-d(a,z)>d(x,p)+9d(p,A)/10-d(p,A)/10.
\end{multline*}
Therefore $d(x,p)+r<d(x,A)$ in this case too.
\end{proof}

We are now ready for formulating and proving the main stability result (denoted by Theorem ~\ref{thm:stabilityUC} in the main body of the text).

\begin{thm}\label{thm:stabilityUC_Appendix}
Let $(\wt{X},|\cdot|)$ be a uniformly convex normed space, and let  $X\subseteq \wt{X}$ be closed and convex. Let $(P_k)_{k\in K}$, $(P'_k)_{k\in K}$ be two given tuples of subsets of  $X$ with the property that the distance between each $x\in X$ and each $P_k,P'_k$ is attained. Suppose that 
\begin{equation*}
\eta:=\inf\{d(P_k,P_j): j,k\in K, j\neq k\}>0, 
\end{equation*}
and  let $A_k=\bigcup_{j\neq k}P_j,\,A'_k=\bigcup_{j\neq k}P'_j$ for each $k\in K$. Suppose also that \eqref{eq:BallRhok} holds. For each $k\in K$ let $R_k=\dom(P_k,A_k),R'_k=\dom(P'_k,A'_k)$ be, respectively, the Voronoi cells associated with the original site $P_k$ and the perturbed one $P'_k$. 
Then for each $\epsilon\in (0,\eta/6)$ there there exists $\Delta>0$, namely $\Delta=\min\{C\epsilon^2, 0.5\epsilon\}$ for any positive $C$ satisfying 
\begin{equation*}%\label{eq:Delta}
C\leq \frac{1}{16(\rho+5\eta/12)}\cdot\delta\left(\frac{\eta}{12\rho+5\eta}
\right),
\end{equation*} 
such that if the inequality  $D(P_k,P'_k)<\Delta$ holds for each $k\in K$, 
then $D(R_k,R'_k)<\epsilon$ for each $k\in K$.
\end{thm}

\begin{proof}
Let $\epsilon\in(0,\eta/6)$ be given. We will show that all the conditions needed in Proposition ~\ref{prop:stabilityVor} are satisfied. Let $k\in K$ be given. As in the first lines of the proof of Proposition \ref{prop:stabilityVor}, we have $D(A_k,A'_k)\leq \sup\{D(P_j,P'_j): j\neq k\}\leq \Delta$, and hence $D(A_k,A'_k)<\epsilon\leq\eta/6$. Let $M=\rho$ and $M'=\rho+\eta/6$ where $\rho$ is from \eqref{eq:BallRhok}. By Lemma \ref{lem:BallRho}\eqref{item:Trho} and Lemma \ref{lem:BallRho}\eqref{item:T'rho} (with $P=P_k$, $P'=P'_k$, $A=A_k$, $A'=A'_k$)  we obtain \eqref{eq:Mk}. 

Let $p'\in P'_k$ be given. We claim that $5\eta/6\leq d(p',A'_k)$.  Indeed, let  $a'\in A'_k$ be arbitrary.  Since $D(P_k,P'_k)<\Delta\leq\epsilon/2$ and $D(A_k,A'_k)<\epsilon/2+\beta$ for any $\beta>0$, there are $p\in P_k$ and $a\in A_k$ such that $d(p,p')<\epsilon/2$ and $d(a',a)<\epsilon/2+\beta$. Therefore 
\begin{equation*}
\eta\leq d(p,A_k)\leq d(p,a)\leq d(p,p')+d(p',a')+d(a',a)< \epsilon+\beta+d(p',a'),
\end{equation*}
 Since $a'$ and $\beta$ were arbitrary and since $5\eta/6\leq \eta-\epsilon$ we have  $5\eta/6\leq d(p',A'_k)$ as claimed. Since $6\epsilon\leq \eta$ we also  conclude that $5\epsilon\leq d(p',A'_k)$. In addition, $0.5\epsilon<2\epsilon<4d(p,A_k)/10$ and $0.5\epsilon<2\epsilon\leq 4d(p',A'_k)/10$ for any $p\in P_k$ and $p'\in P'_k$. 
 
Let $\sigma=\eta/6$. Then obviously $0.5\epsilon<\epsilon\leq\sigma$.  
 By Lemma \ref{lem:BallRho}\eqref{item:RhoAA'} we know that $d(x,A_k)\leq \rho$ and $d(x,A'_k)\leq \rho+\eta/6$ for each $x\in X$. As a result, from  the monotonicity of $\delta$, $0\leq \delta\leq 1$,  $5\eta/6\leq\min\{d(p,A_k),d(p',A'_k)\}$, and from Lemma  \ref{lem:StrictSegment_Appendix}  it follows that 
\begin{equation}\label{eq:lambda_lambda'}
\lambda'=\lambda
=0.5\epsilon\cdot\delta\left(\frac{\eta}{12\rho+5\eta}\right)
\end{equation}
satisfy \eqref{eq:property_k} and \eqref{eq:property'_k} (the minimum in the  expression in the right hand side of \eqref{eq:EpsilonUniform} is the third element and both $\lambda$ and $\lambda'$ are smaller than this element). Let $\epsilon_4=\Delta$. Then $\epsilon_4\leq C\epsilon^2\leq\epsilon\lambda'/(8(M'+\eta/4))<\lambda'/(8(1+M'/\epsilon))$ and hence \eqref{eq:epsilon24M} is satisfied. As a result, we conclude from Proposition \ref{prop:stabilityVor} that $D(R_k,R'_k)<\epsilon$, as required.
\end{proof}

From Theorem \ref{thm:stabilityUC_Appendix} we see that $\Delta=O(\epsilon^2)$. However, under the (practical) additional assumption that the sites are strictly contained in the interior of the world $X$, we can obtain a better bound, namely $\Delta=O(\epsilon)$. The proof is based on the following assertions. We recall that $\partial X$ is the boundary of $X$ relative to the affine hull spanned by $X$ and that the distance from a point to the empty set is infinity. 

\begin{prop}\label{prop:StabilityInterior}
Suppose that $d(P,A)>0$ and let $\epsilon>0$ satisfy  $\epsilon\leq d(P,A)/6$. Suppose that $d(P,\partial X)>0$ and $d(P',\partial X)>0$.  Assume that  \eqref{eq:property} and  \eqref{eq:property'}  hold and that there are two positive numbers $M,M'$ such that \eqref{eq:MM'} hold. Suppose that $\epsilon_1$ and $\epsilon_3$ are two positive numbers such that the following inequalities hold:
\begin{equation}\label{eq:epsilon_kInterior}
\begin{array}{lll}
\epsilon_1+3\epsilon_3 &<& \min\left\{\lambda/2,\lambda'/2\right\}, \\
\epsilon_3&<& \displaystyle{\min\left\{\frac{d(P,\partial X)\epsilon}{4\max\{M,M'\}}, \frac{d(P',\partial X)\epsilon}{4\max\{M,M'\}}\right\}}.
\end{array}
\end{equation}
If 
\begin{equation}\label{eq:HausdorffEpsilonInterior}
D(A,A')< \epsilon_1, \quad
D(P,P')< \epsilon_3,
\end{equation}
 then $D(\dom(P,A),\dom(P',A'))<\epsilon$.
\end{prop}

\begin{proof}
The proof is in the same spirit as the proof of Proposition \ref{prop:stability} and many details are similar and hence they have been omitted here. First, exactly as in the first lines of proof  there we have $5\epsilon\leq d(P',A')$. As explained there, it suffices to show   that  for any $y\in \dom(P,A)$ there exists $y'\in \dom(P',A')$ such that  $d(y,y')<\wt{\epsilon}:=\epsilon-\epsilon_3$. 

Let $y\in \dom(P,A)$. By Theorem \ref{thm:domIntervalApp} we  can write $y=p+t\theta$ for some $p\in P$, a unit vector  $\theta$, and
$t\in [0,T(\theta,p)]$. By \eqref{eq:MM'} we have $t\leq M$. By \eqref{eq:HausdorffEpsilonInterior}  there is $p'\in P'$ such that  $d(p,p')<\epsilon_3$. Suppose first that $t\leq \epsilon$. 
Let $y'=y$. Then $y'\in\dom(P',A')$ exactly as in the proof of Proposition \ref{prop:stability}, and obviously $d(y,y')<\epsilon-\epsilon_3$.

Now assume that $\epsilon<t$.   Let $x=p+(t-0.5\epsilon)\theta$. 
Then $x\in [p,y]$ and $d(x,y)=\epsilon/2$, so by   \eqref{eq:property} we have $d(x,p)\leq d(x,A)-\lambda$. Let $\theta'=\theta$ and  $y'=p'+(t-\epsilon/2)\theta'$. Then $d(x,y')=d(p,p')<\epsilon_3$, and hence $y'\in X$ by Lemma \ref{lem:InX}\eqref{item:B_H(x,r)X} and the second inequality in \eqref{eq:epsilon_kInterior} (in this connection note that since $\lambda<\epsilon$ it follows that  $\epsilon_3<0.5\epsilon\leq\{d(x,y),d(x,p)\}$). In addition, 
\begin{equation*}
d(y,y')\leq d(y,x)+d(x,y')=\epsilon/2+d(p,p')
<\epsilon/2+\epsilon_3<\epsilon-\epsilon_3
\end{equation*}
by \eqref{eq:epsilon_kInterior}. Recalling that $d(y',A)\leq d(y',A')+D(A,A')$ holds in general, we have 
\begin{multline*}%\label{eq:epsilon123}
d(y',p')\leq d(y',x)+d(x,p)+d(p,p')
<\epsilon_3+d(x,A)-\lambda+\epsilon_3\\
\leq 2\epsilon_3+d(x,y')+d(y',A')+D(A,A')-\lambda\\
<3\epsilon_3+\epsilon_1+d(y',A')-\lambda< d(y',A')-\lambda/2,
\end{multline*}
 where we used the first inequality in \eqref{eq:epsilon_kInterior} in the last inequality above. Thus $d(y',p')<d(y',A')$ and then $y'\in \dom(P',A')$. Therefore, in  both cases we have $d(y,y')<\wt{\epsilon}=\epsilon-\epsilon_3$ and $y'\in \dom(P',A')$, and hence  $D(\dom(P,A),\dom(P',A'))<\epsilon$. 
\end{proof}

\begin{prop}\label{prop:stabilityVorInterior}
Suppose that $\inf\{d(P_k,P_j): \, j,k\in K,\, j\neq k\}>0$ and let $\epsilon>0$  be such that $\epsilon\leq \inf\{d(P_k,P_j): \, j,k\in K, k\neq j\}/6$. Suppose also that $d(\bigcup_{k\in K}P_k,\partial X)>0$ and $d(\bigcup_{k\in K}P'_k,\partial X)>0$. Suppose also that  \eqref{eq:property_k} and \eqref{eq:property'_k} hold. Let
\begin{equation*}\label{eq:R_k}
R_k=\dom(P_k,A_k), \quad
\quad R'_k=\dom(P'_k,A'_k).
\end{equation*}
Suppose that there are  $M,M'\in (0,\infty)$ such that for all $k\in K$, \eqref{eq:Mk} hold. Let  $\epsilon_4$ be a positive number satisfying
\begin{equation}\label{eq:epsilon_4Interior}
\epsilon_4<\min\left\{\frac{\lambda}{8},\frac{\lambda'}{8},\frac{d(\bigcup_{k\in K}P_k,\partial X)\epsilon}{4\max\{M,M'\}}, \frac{d(\bigcup_{k\in K}P'_k,\partial X)\epsilon}{4\max\{M,M'\}}\right\}.
\end{equation}
If
\begin{equation}\label{eq:epsilon4kInterior}
 D(P_k,P'_k)< \epsilon_4 \,\,\,\,\forall k\in K,
\end{equation}
 then $D(R_k,R'_k)<\epsilon$ for each $k\in K$.
\end{prop}
\begin{proof}
From the condition $\epsilon\leq d(P_k,P_j)/6$ for each $k\neq j$ it follows that $\epsilon\leq d(P_k,\bigcup_{j\neq k}P_j)/6$. In addition, as in the proof of Proposition \ref{prop:stabilityVor} the assumption $D(P_k,P'_k)<\epsilon_4$ for all $k\in K$ implies that  $D(A_k,A'_k)\leq \epsilon_4$ for all $k\in K$. Let $\epsilon_3:=\epsilon_4$ and  $\epsilon_1:=\epsilon_4+r$, where $r$ is defined to be the minimum of  the two  positive numbers  $(\lambda/2-4\epsilon_4)/4$ and $(\lambda'/2-4\epsilon_4)/4$. Then   \eqref{eq:epsilon_kInterior} is satisfied for $P=P_k$ for each $k\in K$ because of \eqref{eq:epsilon_4Interior}. 
In addition $D(P_k,P'_k)<\epsilon_3$ and  $D(A_k,A'_k)\leq\epsilon_4<\epsilon_1$ for each $k\in K$. Thus all the  conditions of Proposition \ref{prop:StabilityInterior} are satisfied for each $k\in K$. Hence $D(R_k,R'_k)<\epsilon$.
\end{proof}

\begin{thm}\label{thm:stabilityUC_Interior}
Let $(\wt{X},|\cdot|)$ be a uniformly convex normed space, and let  $X\subseteq \wt{X}$ be closed and convex. Let $(P_k)_{k\in K}$, $(P'_k)_{k\in K}$ be two given tuples of subsets of  $X$ with the property that the distance between each $x\in X$ and each $P_k$ and $P'_k$ is attained. Suppose that $d(\bigcup_{k\in K}P_k,\partial X)>0$ and $d(\bigcup_{k\in K}P'_k,\partial X)>0$. For each $k\in K$ let $A_k=\bigcup_{j\neq k}P_j,\,A'_k=\bigcup_{j\neq k}P'_j$ and let $R_k=\dom(P_k,A_k),R'_k=\dom(P'_k,A'_k)$ be, respectively, the Voronoi cells associated with the original site $P_k$ and the perturbed one $P'_k$. Suppose that 
\begin{equation*}
\eta:=\inf\{d(P_k,P_j): j,k\in K, j\neq k\}>0, 
\end{equation*}
and that \eqref{eq:BallRhok} holds.
Then for each $\epsilon>0$ satisfying $\epsilon \leq \min\{\eta/6, 8\cdot d(\bigcup_{k\in K}P_k,\partial X)\}$  there exists $\Delta>0$, namely $\Delta=C\epsilon$ for any $C>0$ satisfying 
\begin{equation*}%\label{eq:Delta}
C\leq \min\left\{\displaystyle{ \frac{1}{16}\delta\left(\frac{\eta}{12\rho+5\eta}\right), \frac{d(\bigcup_{k\in K}P_k,\partial X)}{8(\rho+\eta/6)}}\right\}
\end{equation*} 
such that if the inequality  $D(P_k,P'_k)<\Delta$ holds for each $k\in K$, 
then $D(R_k,R'_k)<\epsilon$ for each $k\in K$.
\end{thm}

\begin{proof}
The proof is almost identical to the proof of Theorem 
\ref{thm:stabilityUC_Appendix} with the exception of a small verification which is explained in the next paragraph. The  values of $\lambda$ and $\lambda'$ are the same those given in \eqref{eq:lambda_lambda'} and one verifies that with $M=\rho$, $M'=\rho+\eta/6$, and $\epsilon_4=\Delta$ the conditions of Proposition \ref{prop:stabilityVorInterior} are satisfied. 

The small verification that mentioned above and should be done is that the inequality 
$\epsilon_4\leq d(\bigcup_{k\in K}P'_k,\partial X)\epsilon/(4\max\{M,M'\})$  holds. This is true for the following reason: if $\partial X=\emptyset$, then the right hand side in this inequality is infinity. Otherwise, let $b\in \partial X$ and $p'\in \bigcup_{k\in K}P'_k$ be given. Then $p'\in P'_k$ for some $k\in K$. Since $\delta\leq 1$ it follows that $C\leq 1/16$ and hence $\epsilon_4\leq \epsilon/16\leq 0.5 d(\bigcup_{k\in K}P_k,\partial X)$.  Because $D(P_k,P'_k)<\epsilon_4$ it follows that $d(p',p)<\epsilon_4$ for some $p\in P_k$ and therefore 
\begin{multline*}
 d(\bigcup_{k\in K}P_k,\partial X)\leq d(p,b)\leq  d(p,p')+d(p',b)<\epsilon_4+d(p',b)\\
 \leq 0.5 d(\bigcup_{k\in K}P_k,\partial X)+d(p',b)
\end{multline*} 
 Since $p'$ and $b$ were arbitrary it follows that $0.5 d(\bigcup_{k\in K}P_k,\partial X)\leq d(\bigcup_{k\in K}P'_k,\partial X)$. Recalling that  $4\max\{M,M'\}=4(\rho+\eta/6)$, we conclude that 
\begin{equation*}
\epsilon_4\leq \frac{0.5\cdot d(\bigcup_{k\in K}P_k,\partial X)\epsilon}{4(\rho+\eta/6)}\leq \frac{d(\bigcup_{k\in K}P'_k,\partial X)\epsilon}{4(\rho+\eta/6)}.
\end{equation*} 
\end{proof}

\section{Appendix 2: A continuity property of $T$}\label{sec:appendix_T}
This appendix is devoted to a discussion on  the continuity of $T(\cdot,p)$ in uniformly convex spaces which was mentioned briefly in Section \ref{sec:Concluding}. This result is related to the main  result and the approach used for deriving it is similar. 

The following definition plays an important role.
\begin{defin}\label{def:emanation}
Let $X$ be a closed and  convex subset of a normed space. Let $p\in X$. Let $\Theta_p:=\{\theta: |\theta|=1,\,p+t\theta\in X\,\,\textnormal{for some}\,\, t>0\}$ be the set of all directions such that rays emanating from $p$ in these directions intersect $X$ not only in $p$.  Let $\theta\in\Theta_p$. Let $L(\theta)\in (0,\infty]$ be the length of the 
line segment generated from the intersection of $X$ and the ray emanating from $p$  in the direction of $\theta$. The point $p$ is  said to have the emanation property (or to satisfy the  emanation condition)  in the direction of $\theta$ if for each $\epsilon>0$ there exists $\beta>0$ such that for any $\phi\in\Theta_p$, if $|\phi-\theta|<\beta$, then the intersection of $X$ and the ray emanating from $p$  in the direction of $\phi$ is a line segment of length at least $L(\theta)-\epsilon$. 
In other words, $L(\phi)\geq L(\theta)-\epsilon$. The point $p$ is said to have  the emanation property if it has the emanation property in the direction of every $\theta\in\Theta_p$. A subset $C$ of $X$ is said to have  the emanation property if each $p\in C$ has the emanation property. 
\end{defin}

As an illustration, in the following examples each point $p$ in the set $X$  satisfies the  emanation condition. First example: $X$ is any bounded closed convex set and $p\in X$ is an arbitrary point in the interior of $X$ relative to the affine hull spanned by $X$; second example: the boundary of the bounded closed and  convex $X$ is strictly convex (if $a\neq b$ are two points in the boundary, then the open line segment $(a,b)$ is contained in the interior of $X$ relative to the affine hull spanned by $X$) and $p\in X$ is arbitrary. Any ball in a strictly convex space has a strictly convex boundary; third example:  $X$ is a cube (of any finite dimension) and $p\in X$ is arbitrary; fourth example: $X$ is a closed linear subspace and $p\in X$  is arbitrary. The proof of the third example  is by direct computation. The proof of the first and  the second examples  is based on Lemma \ref{lem:EmanationInterior} below  and, e.g., on Lemma \ref{lem:InX}\eqref{item:B_H(x,r)X}. 

Before formulating Lemma \ref{lem:EmanationInterior} we note that the reason we consider only $\Theta_p$ and not the whole unit sphere is that in general it may happen that $L(\phi)=0$ for unit vectors $\phi$ arbitrary close to $\theta$ (consider e.g., $X=[0,1]^2$ in the Euclidean plane, $p=(0,0)$,  $\theta=(1,0)$, and $\phi_n=(\sqrt{1-1/n^2},-1/n),\,n\in\N$ or $X=[0,1]^2\times\{0\}$ in $\R^3$ with the Euclidean norm, $p=(0,0,0)$,  $\theta=(1,0,0)$, and $\phi_n=(\sqrt{1-1/n^2},0,-1/n),\,n\in\N$) and hence the condition $L(\phi)\geq L(\theta)-\epsilon$ does not hold. Intuitively, rays in the  directions of  $\phi\notin \Theta_p$  go out $X$ immediately and they do not affect the ``events''  inside $X$. 

\begin{lem}\label{lem:EmanationInterior}
 Let $X$ be any closed and convex subset of a normed space, let $p\in X$, and let $\theta\in \Theta_p$ be given. Assume that $L(\theta)<\infty$ and that any point in the segment  $(p,p+L(\theta)\theta)$ is in the interior of $X$ relative to the affine hull $H$ spanned by $X$. Then $p$ satisfies the  emanation  condition in the direction of $\theta$. 
\end{lem}
\begin{proof}
$L(\theta)>0$ since $\theta\in \Theta_p$. Let $\epsilon\in (0,L(\theta))$ be given. Let $z=p+(L(\theta)-\epsilon)\theta$. Then $z$ is well defined and it is in the interior of the segment $[p,p+L(\theta)\theta]$ and hence there exists $\rho>0$ such that the intersection of the ball $B(z,\rho)$ and the affine hull $H$ is contained in $X$. Let $r=\rho/(L(\theta)-\epsilon)$. For  each $\phi\in \Theta_p$, if $|\phi-\theta|<r$, then $|(p+(L(\theta)-\epsilon)\phi)-z|=(L(\theta)-\epsilon)|\phi-\theta|<\rho$.  Thus the point $p+(L(\theta)-\epsilon)\phi$, which is in $H$ since $\phi\in \Theta_p$, is also in the ball $B(z,\rho)$ and hence in $X$. Thus $L(\phi)\geq L(\theta)-\epsilon$ by the definition of $L(\phi)$.
\end{proof}
\begin{expl}\label{ex:NonEmanation}
The emanation condition does not hold in general. Consider the Hilbert space $\ell_2$. Let $(e_n)_{n=1}^{\infty}$ be the standard basis. Let $y_1=e_1$ and for each  $n>1$ let $y_n=e_1/2+e_n/n$. Let $A=\{-e_1\}$.  Let $X$ be the closed convex hull generated by $A\bigcup\{y_n: n=1,2,\ldots\}$. Let $p=0$ and let  $\theta_n=y_n/|y_n|$ for each $n$. Then $p$ does not satisfy the emanation  condition in the direction of $\theta_1$. Indeed, $\lim_{n\to\infty}\theta_n=\theta_1$ but $L(\theta_n)=\sqrt{0.25+1/n^2}<0.99=L(\theta_1)-0.01$ for each $n>1$. The  subset $X$ is in fact compact since the sequence $(y_n)_{n=1}^{\infty}$  converges. 
\end{expl}

Using the definition of the emanation condition, we can prove the continuity of $T(\cdot,p)$. The proof is based on several lemmas.
\begin{lem}\label{lem:y_1y_2}
Let $X$ be a convex subset of a normed space. Let $p\in X$ and $\emptyset\neq A\subset X$ be given. Suppose that $d(p,A)>0$.  Let $\epsilon\in (0,d(p,A)/6)$ and assume that the following conditions hold:
\begin{multline}\label{eq:propertyT}
\textnormal{There exists}\,\, \lambda\in (0,\epsilon)\,\, \textnormal{such that for each}\,\, y\in \dom(p,A),\,\, x\in [p,y],\\
 \textnormal{if}\,\, d(x,y)=\epsilon/2,\,\, \textnormal{then}\,\,
 d(x,p)\leq d(x,A)-\lambda,
\end{multline} 
\begin{equation}\label{eq:MT}
\textnormal{there exists}\,\, M\in (0,\infty)\,\, \textnormal{such that}
\,\,\,\sup\{T(\theta,p): |\theta|=1\}\leq M.
\end{equation}
Let $\beta$ be a positive number satisfying $\beta\leq \lambda/(4M)$. Let $y_i=p+t_i\theta_i,\,i=1,2$ where  
$t_1\geq t_2+0.5\epsilon$, $t_2\geq 0$, and $\theta_i\in \Theta_p,\,i=1,2$. If $y_1\in \dom(p,A)$ and $|\theta_1-\theta_2|<\beta$, then $d(y_2,p)\leq d(y_2,A)$. 
\end{lem}
\begin{proof}
Suppose first that $t_1\leq \epsilon$. Then 
\begin{equation*}
6\epsilon\leq d(p,A)\leq d(p,y_2)+d(y_2,A)\leq\epsilon+d(y_2,A)
\end{equation*}
since $t_2\leq t_1\leq \epsilon$. Therefore $d(y_2,p)\leq\epsilon<5\epsilon\leq d(y_2,A)$ as claimed. 

Now consider the case where $\epsilon<t_1$. Let $x=p+(t_1-\epsilon/2)\theta_1$ and $z=p+(t_1-\epsilon/2)\theta_2$. Note that $z$ and also $y_2$ are not necessarily in $X$, but this is not important.  By \eqref{eq:MT} we have $t_1\leq M$. Hence  $d(x,z)=(t_1-\epsilon/2)|\theta_1-\theta_2|<M\beta$.  Since  $d(x,y_1)=\epsilon/2$, $y_1\in\dom(p,A)$,  and $x\in [p,y_1]$, it follows from \eqref{eq:propertyT} and the choice of $\beta$ that
\begin{multline*}
d(z,p)\leq d(z,x)+d(x,p)< M\beta+d(x,A)-\lambda\\
< M\beta+d(x,z)+d(z,A)-\lambda\leq d(z,A)-\lambda/2.
\end{multline*}
Since $t_2\leq t_1-0.5\epsilon$ it follows that $y_2\in [p,z]$. As a result,  Lemma \ref{lem:segment}  implies that  $d(y_2,p)\leq d(y_2,A)$. 
\end{proof}

\begin{lem}\label{lem:T_UppCont}
Let $X$ be a closed and convex subset of a normed space. Let $A\subseteq X$. Let $p\in X$ and let $\theta\in \Theta_p$. Suppose that there exists $M\in(0,\infty)$ such that $T(\phi,p)\in [0,M]$ for all $\phi\in \Theta_p$, where $T(\phi,p)$ is defined with respect to $A$ in \eqref{eq:Tdef}. Then for any $\epsilon\in (0,\infty)$ there exists $\beta>0$ such that  for all $\phi\in \Theta_p$, if  $|\theta-\phi|<\beta$, then $T(\phi,p)<T(\theta,p)+\epsilon$.  
\end{lem}
\begin{proof}
Suppose that this is not true. Then for some $\epsilon\in (0,\infty)$ there exists a sequence $(\phi_n)_{n=1}^{\infty}$ of elements in $\Theta_p$ such that $\theta=\lim_{n\to\infty}\phi_n$ but $T(\phi_n,p)\geq T(\theta,p)+\epsilon$. By assumption $T(\phi_n,p)\in [T(\theta,p)+\epsilon,M+\epsilon]$ for all $n$. Therefore there exists a subsequence $(\phi_{n_j})_{j=1}^{\infty}$ such that $t:=\lim_{j\to\infty}T(\phi_{n_j},p)$ exists and satisfies $t\in [T(\theta,p)+\epsilon,M+\epsilon]$. Hence $T(\phi_{n_j},p)\phi_{n_j}\to t\theta$ and  $\lim_{j\to\infty}(p+T(\phi_{n_j},p)\phi_{n_j})=p+t\theta$. Because $X$ is closed it follows that $p+t\theta\in X$. By Theorem \ref{thm:domIntervalApp} (with $P=\{p\}$) we know that $p+T(\phi_{n_j},p)\phi_{n_j}\in \dom(p,A)$ for each $j$. Therefore the inequality   $d(p+T(\phi_{n_j})\phi_{n_j},p)-d(p+T(\phi_{n_j},p)\phi_{n_j},A)\leq 0$ holds  for each $j$. Thus, from the continuity of the distance function it follows  that $d(p+t\theta,p)-d(p+t\theta,A)\leq 0$. As a result $p+t\theta\in \dom(p,A)$. 
But $T(\theta,p)<t$  since $T(\theta,p)\leq t-\epsilon$, and this is a contradiction to the definition of $T(\theta,p)$ (see \eqref{eq:Tdef}). 
\end{proof}

\begin{thm}\label{thm:contT}
Let $X$ be a closed and convex subset of a uniformly convex
normed space and let $A$ be a subset of $X$. Let $p\in X$ and 
suppose that  $\eta:=d(p,A)>0$. Suppose that \eqref{eq:BallRho} holds and that $p$  has the emanation property. Then for each $\epsilon \in 
(0,\eta/6)$ and each $\theta \in \Theta_p$ there exists 
$\Delta \in (0,\infty)$ such that for each
 $\phi \in \Theta_p$, if $|\theta-\phi|<\Delta$, then
$|T(\theta,p)-T(\phi,p)|\leq \epsilon$.
\end{thm}

\begin{proof}
Let $\sigma=\eta/6$ and 
\begin{equation*}
\lambda=0.5\epsilon\cdot\delta\left(\frac{\eta}{10(\rho+\eta/4)}\right).
\end{equation*}
Then $\lambda$  satisfies \eqref{eq:propertyT}. Indeed, let $x$ and $y$ be   such that  $y\in \dom(p,A)$, $x\in [p,y]$, and $d(x,y)=0.5\epsilon$. We have $d(x,A)<\rho$ by  Lemma \ref{lem:BallRho}\eqref{item:RhoAA'} since  \eqref{eq:BallRho} holds. In addition,  
$0<\lambda\leq 0.5\epsilon<\sigma<4\eta/10$ and $\lambda$ is not greater  than the minimum in \eqref{eq:EpsilonUniform}. From Lemma \ref{lem:StrictSegment_Appendix} it follows that $d(x,p)<d(x,A)-\lambda$. 

From Lemma \ref{lem:T_UppCont} there exists $\beta_1>0$ such that for all $\phi\in \Theta_p$, if  $|\theta-\phi|<\beta_1$, then $T(\phi,p)<T(\theta,p)+\epsilon$. Since $p$ has the emanation property there exists $\beta_2>0$ such that  for all $\phi\in \Theta_p$, if  $|\theta-\phi|<\beta_2$, then $L(\phi)\geq L(\theta)-\epsilon$. Let 
\begin{equation*}
\Delta=\min\{\beta_1,\beta_2,\lambda/(4\rho)\}. 
\end{equation*}
We claim that for each $\phi \in \Theta_p$, if $|\theta-\phi|<\Delta$, then
$|T(\theta,p)-T(\phi,p)|\leq\epsilon$. Indeed, let $\phi\in \Theta_p$ satisfy $|\theta-\phi|<\Delta$. Then $|\theta-\phi|<\beta_1$, so 
Lemma \ref{lem:T_UppCont} implies that $T(\phi,p)<T(\theta,p)+\epsilon$. 

Assume for contradiction that $T(\theta,p)> T(\phi,p)+\epsilon$.  
Let $t_1=T(\theta,p)$, $y_1=p+t_1\theta$, $t_2=t_1-\epsilon$,  and $y_2=p+t_2\phi$. In particular $t_2>T(\phi,p)$. Because $|\theta-\phi|<\beta_2$  it follows that $L(\phi)\geq L(\theta)-\epsilon\geq T(\theta,p)-\epsilon=t_2$. Therefore the point $y_2$, which is on the ray emanating from $p$ in the direction of $\phi$, is in $X$. By Lemma ~\ref{lem:BallRho}\eqref{item:Trho} (with $P=\{p\}$) we know that \eqref{eq:MT} is satisfied with $M=\rho$. In addition, since $|\theta-\phi|<\lambda/(4\rho)$ and since $y_1\in X$ and $d(y_1,p)\leq d(y_1,A)$ (by Theorem \ref{thm:domIntervalApp}), it follows from Lemma \ref{lem:y_1y_2}  that $d(y_2,p)\leq d(y_2,A)$. By the definition of $T(\phi,p)$  we therefore have $t_2\leq T(\phi,p)$, a contradiction.
\end{proof}

\begin{remark}\label{rem:CounterExampleT}
The conditions mentioned in Theorem \ref{thm:contT} are necessary for the continuity of $T$ mentioned there. Indeed, the uniform convexity of the norm is needed as can be seen  in the case described in Figure \ref{fig:InStability000} where $p=(0,0)$ and $A=\{(0,-2), (2,0), (-2,0)\}$. The discontinuity of $T(\cdot,p)$ is at the four unit vectors $\theta=(\pm 1,\pm 1)$. The emanation property is necessary since in  Example \ref{ex:NonEmanation} it can be seen that $T(\cdot,p)$, $p=0$, is  discontinuous at $\theta=\theta_1$. As for the assumption $d(p,A)>0$, consider $X=[-1,1]^2$ in the Euclidean plane where  $p=(0,0)$ and $A=[-1,1]\times \{0\}$. Then $T(\theta,p)$ is discontinuous at $\theta=(0,1)$ because $T(\theta,p)=1$ but $T(\phi,p)=0$ for any $\phi$ of the form $\phi=(1/n,\sqrt{1-1/n^2})$, $n\in \N$. As for the boundedness condition \eqref{eq:BallRho}, let $X=\R^2$ with the Euclidean norm, $p=(0,-1)$, and $A=\{(0,1)\}$. Then $T(\cdot,p)$ is not continuous at $\theta=(1,0)$, because $T(\theta,p)=\infty$ but $T(\phi,p)<\infty$ for any $\phi=(\sqrt{1-1/n^2},1/n),\,n\in\N$.   
\end{remark}

\end{document}